\documentclass[letterpaper,onecolumn,draftclsnofoot]{IEEEtran}
\usepackage{amsmath,amsfonts,amssymb,amsthm}
\usepackage{ascmac}
\usepackage{bm}
\usepackage{algorithmic}
\usepackage{algorithm}
\usepackage{array}
\usepackage{hhline}
\usepackage[caption=false,font=normalsize,labelfont=sf,textfont=sf]{subfig}
\usepackage{tabularx}
\usepackage{here}
\usepackage{multicol}
\usepackage{multirow}
\usepackage{longtable}
\usepackage{textcomp}
\usepackage{stfloats}
\usepackage{url}
\usepackage{verbatim}
\usepackage{graphicx}
\usepackage[table]{xcolor}
\usepackage{colortbl}
\usepackage{cite}
\newtheorem{theorem}{Theorem}[section]
\newtheorem{lemma}[theorem]{Lemma}
\newtheorem{corollary}[theorem]{Corollary}
\newtheorem{example}[theorem]{Example}
\newtheorem{definition}[theorem]{Definition}
\newtheorem{proposition}[theorem]{Proposition}

\newtheorem{construction}[theorem]{Construction}

\counterwithin{table}{section}

\allowdisplaybreaks

\def\bF{\mathbb{F}}
\def\bN{\mathbb{N}}

\def\bZ{\mathbb{Z}}

\def\A{\mathcal{A}}
\def\D{\mathcal{D}}

\DeclareMathOperator{\Enc}{Enc}
\DeclareMathOperator{\Dec}{Dec}

\newcommand{\rhoEE}[1]{\rho_{{\scriptscriptstyle \mathrm{E}}, #1}}

\newcommand{\rhoEstr}{\rho_{{\scriptscriptstyle \mathrm{E}}}^{\mathrm{str}}}
\newcommand{\rhoEwk}{\rho_{{\scriptscriptstyle \mathrm{E}}}^{\mathrm{wk}}}

\newcommand{\rhoEstrS}{\rho_{{\scriptscriptstyle \mathrm{E}}, s}^{\mathrm{str}}}

\begin{document}

\title{Cyclotomy, External Difference Families, and Algebraic Manipulation Detection Codes}

\author{Minfeng Shao,
Miwako Mishima\textsuperscript{*}
\thanks{\textsuperscript{*}Corresponding author.}
\thanks{M. Shao is with the School of Computer and Software Engineering, Xihua University, Chengdu, China. E-mail: shaomf@mail.xhu.edu.cn.}
\thanks{M. Mishima is with the Department of Electrical, Electronic and Computer Engineering, Gifu University, Gifu, Japan. E-mail: mishima.miwako.n0@f.gifu-u.ac.jp.}
\thanks{This work was supported by JSPS KAKENHI Grant Numbers
JP24K14819, JP24K06834, and JP22K11936 (Miwako Mishima).}}



\maketitle

\begin{abstract}
External difference families (EDFs) are closely related to
algebraic manipulation detection (AMD) codes, a cryptographic
primitive that protects messages against additive tampering by an
adversary who cannot observe the transmitted codeword.
We first study a cyclotomic construction of external partial
difference families over finite fields and derive a necessary and
sufficient condition under which the resulting families are EDFs.
For even block sizes up to $14$, 
we obtain explicit criteria in
terms of quadratic and biquadratic residues, yielding new
$R$-optimal weak AMD codes.
We then construct bounded generalized strong external difference
families using cyclotomic classes over finite fields, direct
products of finite fields, and generalized cyclotomic classes over
integer residue rings.
These constructions give infinite families of systematic
$G$-optimal strong AMD codes with flexible parameters.
\end{abstract}


\section{Introduction}
\IEEEPARstart{T}{o} convert linear secret sharing schemes into robust secret sharing schemes, 
Cramer et al.~\cite{cramer2008detection} introduced algebraic manipulation
detection (AMD) codes, which have also found applications in the
construction of robust fuzzy extractors.
As powerful tools against algebraic manipulation
attacks, AMD codes have attracted significant attention in recent years.
Among the two standard variants of AMD codes, namely weak AMD codes and strong
AMD codes, this paper investigates constructions of both types.

The literature on AMD codes has mainly developed in two directions.
One direction concerns combinatorial characterizations of AMD codes.
Cramer et al.~\cite{cramer2008detection} showed that AMD codes are closely
related to certain differential structures. 
Later, Paterson and Stinson~\cite{paterson2016combinatorial} proved that R-optimal AMD codes for the strong model are equivalent to generalized strong external difference families. 
They also introduced bounded generalized strong external difference families, which are closely related to G-optimal AMD codes for the
strong model. 
For the weak model, Shao and Miao~\cite{shao2020optimal} showed that R-optimal AMD codes are equivalent to strong weighted external difference families.

The second direction concerns constructions of AMD codes. Cramer
et al.~\cite{cramer2008detection} introduced a flexible construction based on
Reed--Solomon codes. Subsequent constructions employed various algebraic and
combinatorial techniques, including cyclic codes such as BCH
codes~\cite{cramer2015optimal}, Reed--Muller
codes~\cite{karpovsky2013design}, and differential
structures~\cite{cramer2013algebraic}. Since AMD codes are closely related to
various types of external difference families, combinatorial constructions of
these families also play an important role in constructing AMD codes.
Paterson and Stinson~\cite{paterson2016combinatorial} introduced strong external difference families to obtain AMD codes in various settings, while for the weak attack model, 
Shao and Miao~\cite{shao2020optimal} employed weighted strong
external difference families to construct AMD codes with new parameters.

Existing constructions of AMD codes mainly focus on achieving optimality with respect to the R-bound. 
In contrast, much less is known about G-optimal AMD codes, which attain the G-bound.
To the best of our knowledge, explicit constructions of G-optimal AMD codes have not been previously available.

This paper investigates constructions of both R-optimal and G-optimal AMD codes via external difference families and bounded generalized strong external difference families. Our first construction is based on a class of EPDFs introduced by Huczynska and Johnson~\cite{huczynska2023internal}. 
They gave necessary and sufficient conditions under which an EPDF, obtained by partitioning quadratic residues into cyclotomic classes, forms an EDF.
However, these conditions are intricate and difficult to verify in practice.
Focusing on EPDFs with even block sizes up to $14$ over finite fields $\bF_p$, where $p \equiv 1 \pmod{4}$ is prime, we derive simplified criteria by analyzing quadratic and biquadratic residues. The resulting conditions are substantially easier to verify and make it practical to determine whether a given EPDF forms an EDF. For reference, we also list, for each block size,
the first several primes $p$ for which the corresponding EPDF satisfies these criteria.

In contrast, we investigate G-optimal AMD codes through bounded generalized strong external difference families.
We first present constructions based on cyclotomy over finite fields.
Although these constructions yield optimal AMD codes, their parameters are restricted by the sizes of cyclotomic classes, which must divide $q-1$ for a prime power $q$.
We also extend the finite-field construction to direct products of finite fields, obtaining additional infinite families of bounded generalized strong external difference families and G-optimal AMD codes.
Finally, we introduce a construction based on the Zeng--Cai--Tang--Yang generalized cyclotomy over residue class rings~\cite{zeng2013optimal}, leading to bounded generalized strong external difference families and optimal AMD codes with more flexible parameters.

The remainder of this paper is organized as follows.
Sections~\ref{sec:preliminary} and~\ref{sec-BSWEDF} provide preliminaries on AMD codes and related external difference families.
Section~\ref{miwakoEPDF} presents constructions of AMD codes from EPDFs, along with criteria under which these EPDFs yield EDFs. We also apply these results to small parameters to obtain explicit constructions and examples.
Section~\ref{sec_cons_cosets} presents constructions of bounded generalized strong external difference families and the corresponding systematic AMD codes based on finite fields, direct products of finite fields, and residue class rings of integers. Section \ref{sec-conclusion} concludes the paper with some remarks.

\section{Algebraic manipulation detection codes}\label{sec:preliminary}

This section introduces notation and preliminaries on algebraic manipulation detection (AMD) codes.

\begin{itemize}
\item $[n] := \{1,2,\ldots,n\}$ for a positive integer $n$;

\item $\bZ_n$: the residue class ring modulo $n$;

\item $\bF_q$: the finite field of order $q$, and
$\bF_q^*:=\bF_q\setminus\{0\}$;

\item $G^*:=G\setminus\{0\}$ for an Abelian group $(G,+)$
with identity element $0$;

\item $c\cdot B$: the multiset in which each element of $B$
appears with multiplicity $c$;

\item $aB:=\{ab:b\in B\}$ for $a\in\bZ_n$ and
$B\subseteq\bZ_n$;

\item $AB:=\{ab:a\in A,\ b\in B\}$ for subsets
$A,B\subseteq\bF_q$, with multiplicities counted;

\item $D(B):=\{b-b':b,b'\in B,\ b\neq b'\}$ for
$B\subseteq G$;

\item $D(B_1,B_2):=\{b_1-b_2:b_1\in B_1,\ b_2\in B_2\}$
for distinct subsets $B_1,B_2\subseteq G$;

\item $s\in_R S$: uniformly random choice of $s$ from $S$.
\end{itemize}

Let $S$ be a source space with $|S|=m$, and let $G$ be an Abelian group of
order $n$ with identity element $0$. For each source $s \in S$, let
$A_s \subseteq G$ denote the set of valid encodings of $s$. 
A probabilistic encoding function $\Enc$ (which may be denoted by $E$ in subscripts for simplicity) maps each source $s \in S$ to an element of $A_s$ according to a prescribed probability distribution. 
The sets $A_s$ are assumed
to be pairwise disjoint, so that decoding is uniquely determined by a decoding
function $\Dec : G \to S \cup \{\perp\}$. Denote
\[
\A := \{A_s : s \in S\}.
\]
Then $(S,G,\A,\Enc)$ is called an \emph{AMD code}.

An adversary is assumed to have complete knowledge of the AMD code and chooses
a nonzero value $\Delta \in G^*$ according to a strategy
$\sigma \in \Sigma$, where $\Delta$ represents the difference between the
tampered and original messages. Given $s \in S$, tampering is successful if
and only if $\Dec(\Enc(s)+\Delta)=s' \in S$ for some $s'\neq s$. The success
probability of tampering is
\begin{align*}
&\Pr[\Dec(\Enc(s)+\Delta)\notin\{s,\bot\}\mid s\in S,\Delta\in G^*] \\
&=\sum_{g\in A_s}\Pr[\Enc(s)=g]
\Pr[\Dec(g+\Delta)\notin\{s,\bot\}\mid s\in S, \Delta\in G^*],
\end{align*}
which is denoted by $\rhoEE{s,\sigma}$ for the strategy
$\sigma \in \Sigma$ and source $s\in S$. For a fixed source $s\in S$, define
\[
\rhoEE{s}:=\max_{\sigma\in\Sigma}\rhoEE{s,\sigma}.
\]
Let $\rhoEstr$, where ``str'' indicates the strong security model, denote the
maximum success probability when the adversary knows the source $s$, i.e.,
\[
\rhoEstr
=\max_{s\in S,\ \sigma\in\Sigma}\rhoEE{s,\sigma}
=\max_{s\in S}\rhoEE{s}.
\]

If the adversary chooses $\Delta$ before observing $s$, then
\[
\rhoEE{\sigma}
:=\sum_{s\in S}\Pr[s]\,\rhoEE{s,\sigma}
\]
is the corresponding success probability for the strategy $\sigma$. 
Let $\rhoEwk$, where ``wk'' indicates the weak security model, denote the
maximum success probability when the adversary chooses $\Delta$ before
observing $s$, i.e.,
\[
\rhoEwk
=\max_{\sigma\in\Sigma}\rhoEE{\sigma}
=\max_{\sigma\in\Sigma}\sum_{s\in S}\Pr[s]\,\rhoEE{s,\sigma}.
\]
It is easy to verify that $\rhoEwk\le\rhoEstr$ for any AMD code.

\begin{definition}\label{def_AMD}
Suppose $(S,G,\A,\Enc)$ is an AMD code with $|S|=m$ and $|G|=n$.
Let $K:=\{|A_s|:s\in S\}$. 

\begin{itemize}
\item[(1)]
$(S,G,\A,\Enc)$ is a strong
$(m,n,K,\rhoEstr)$-AMD code if
\[
\Pr[\Dec(\Enc(s)+\Delta)\notin\{s,\perp\}\mid s\in S, \Delta\in G^*]
\le \rhoEstr <1
\]
for all $s\in S$ and $\Delta\in G^*$.

\item[(2)]
$(S,G,\A,\Enc)$ is a weak
$(m,n,K,\rhoEwk)$-AMD code if
\[
\sum_{s\in S}\Pr[s]\,
\Pr[\Dec(\Enc(s)+\Delta)\notin\{s,\perp\}\mid s\in S, \Delta\in G^*]
\le \rhoEwk <1
\]
for all $\Delta\in G^*$.

\item[(3)]
$(S,G,\A,\Enc)$ is called \emph{systematic}
if there exist Abelian groups $G_1$ and $G_2$ such that
$G=S\times G_1\times G_2$, and for some function
$f:S\times G_1\to G_2$, the encoding has the form
\begin{equation}
\Enc(s)=(s,x,f(s,x)),
\qquad x\in_R G_1.
\end{equation}
\end{itemize}
\end{definition}

We represent the multiset $K$ using exponential notation as
$k_1^{u_1}k_2^{u_2} \cdots k_t^{u_t}$, 
where $k_i$ has multiplicity $u_i$.
An AMD code is called \emph{$k$-uniform} if $|A_s|=k$ for every $s\in S$.
A $k$-uniform AMD code with equiprobable sources and equiprobable
encoding is said to be \emph{$k$-regular}.
Unless otherwise stated, we assume equiprobable sources and equiprobable
encoding, namely, $\Pr[s]=1/m$ and
$\Pr[\Enc(s)=g]=1/|A_s|$ for every $s\in S$ and $g\in A_s$.

In general, for a given source space $S$, one seeks AMD codes with small
length $n$ and low tampering probability. However, these parameters are not
independent. For weak AMD codes, the following bounds are known.

\begin{theorem}[\cite{paterson2016combinatorial}]
\label{thm_weak_R_bound}
(Weak model $R$-bound)
For any weak $(m,n,K,\rhoEwk)$-AMD code
$(S,G,\A,\Enc)$ with equiprobable sources,
the tampering probability $\rhoEwk$ under the optimal random strategy
satisfies
\begin{equation*}
\rhoEwk \ge \frac{a(m-1)}{m(n-1)},
\end{equation*}
where $a:=\sum_{s\in S}|A_s|$.
Furthermore, if the AMD code is $k$-uniform, then
\begin{equation*}
\rhoEwk \ge \frac{k(m-1)}{n-1}.
\end{equation*}
\end{theorem}

\begin{theorem}[\cite{paterson2016combinatorial}]
\label{thm_weak_G_bound}
(Weak model $G$-bound)
For any weak $(m,n,K,\rhoEwk)$-AMD code
$(S,G,\A,\Enc)$, the tampering probability
$\rhoEwk$ under the optimal guessing strategy satisfies
\begin{equation*}
\rhoEwk \ge \frac{1}{a},
\end{equation*}
where $a:=\sum_{s\in S}|A_s|$.
\end{theorem}

Weak AMD codes attaining equality in these bounds are of particular interest.

\begin{definition}
\label{def_R_op_PS}
A weak AMD code is said to be
\emph{$R$-optimal} (resp.\ \emph{$G$-optimal})
if it attains the bound of
Theorem~\ref{thm_weak_R_bound}
(resp.\ Theorem~\ref{thm_weak_G_bound})
with equality.
\end{definition}

Analogous bounds hold for strong AMD codes.

\begin{theorem}[\cite{paterson2016combinatorial}]
\label{thm_strong_R_bound}
(Strong model $R$-bound)
For any strong $(m,n,K,\rhoEstr)$-AMD code
$(S,G,\A,\Enc)$, the tampering probability
$\rhoEstrS$ under the optimal random strategy satisfies
\begin{equation*}
\rhoEstrS \ge \frac{a-|A_s|}{n-1},
\end{equation*}
where $a:=\sum_{s'\in S}|A_{s'}|$.
Furthermore, if the AMD code is $k$-uniform, then
\begin{equation*}
\rhoEstr \ge \frac{k(m-1)}{n-1}.
\end{equation*}
\end{theorem}

\begin{theorem}[\cite{paterson2016combinatorial}]
\label{thm_strong_G_bound}
(Strong model $G$-bound)
For any strong $(m,n,K,\rhoEstr)$-AMD code
$(S,G,\A,\Enc)$, the tampering probability
$\rhoEstrS$ under the optimal guessing strategy satisfies
\begin{equation*}
\rhoEstrS \ge \frac{1}{|A_s|}.
\end{equation*}
Furthermore, if the AMD code is $k$-uniform, then
\begin{equation*}
\rhoEstr \ge \frac{1}{k}.
\end{equation*}
\end{theorem}

For systematic strong AMD codes, the following stronger bound is known.

\begin{theorem}[\cite{cramer2013algebraic,paterson2016combinatorial}]
\label{thm_sys_bound}
Let $(S,G,\A,\Enc)$ be a systematic strong
$(m,n=mn_1n_2,K,\rhoEstr)$-AMD code, where
$G=S\times G_1\times G_2$ with $|G_1|=n_1$ and
$|G_2|=n_2$. Then the tampering probability
$\rhoEstr$ under the optimal guessing strategy satisfies
\begin{equation}
\label{eqn_bound_sys_G}
\rhoEstr \ge \frac{1}{n_1}.
\end{equation}
\end{theorem}

The notions of $R$-optimality and $G$-optimality for strong AMD codes are defined analogously.

\begin{definition}
A strong AMD code is said to be
\emph{$R$-optimal} (resp.\ \emph{$G$-optimal})
if it attains the bound of
Theorem~\ref{thm_strong_R_bound}
(resp.\ Theorem~\ref{thm_strong_G_bound}
or Theorem~\ref{thm_sys_bound})
with equality for every source $s\in S$.
\end{definition}

\section{External difference families and combinatorial tools}
\label{sec-BSWEDF}

The construction of optimal AMD codes is closely related to
external difference families through the set system
$\A=\{A_s:s\in S\}$. In particular, several classes of optimal
AMD codes can be characterized in terms of suitable difference
properties among the sets $A_s$.

To develop the constructions used in later sections, we introduce
the necessary notions on external difference families and cyclotomy.
Cyclotomic classes provide a convenient way to partition finite fields into structured subsets, which will play a central role in our constructions.

We begin with notation for cyclotomy.
Let $e,u\in\bN$ such that $q=eu+1$ is a prime power.
Let $\bF_q^*$ be the multiplicative group of the finite field
$\bF_q$, and let $\gamma$ be a primitive element of $\bF_q$.
The \emph{cyclotomic classes} of index $u$ over $\bF_q$ are defined by
\begin{equation*}
C_0^u
\triangleq
\{\gamma^{ju}:0\le j\le e-1\}
\end{equation*}
and
\begin{equation*}
C_i^u
=
\gamma^i C_0^u
\triangleq
\{\gamma^{ju+i}:0\le j\le e-1\}
\text{ for }0\le i\le u-1.
\end{equation*}
It is well known that
$\{C_i^u:0\le i\le u-1\}$
forms a partition of $\bF_q^*$ into $u$ subsets of size $e$.

Throughout this paper, the union symbol $\cup$ denotes multiset
union unless otherwise specified. 

\begin{definition}[Difference Set] 
\label{def_DS}
Let $G$ be an additive Abelian group of order $n$. 
A subset $B \subseteq G$ of size $k$ is called an \emph{$(n, k, \lambda)$-difference set} (DS) if its internal difference list satisfies
\[
    D(B) = \lambda \cdot G^*,
\]
where $G^* = G \setminus \{0\}$.
\end{definition}

If an $(n,k,\lambda)$-DS exists, then its parameters satisfy
\[
\lambda(n-1)=k(k-1).
\]

More generally, one may consider the following relaxation.

\begin{definition}[Partial Difference Set] 
\label{def_PDS}
Let $G$ be an additive Abelian group of order $n$. 
A subset $B \subseteq G$ of size $k$ is called an \emph{$(n, k, \lambda, \mu)$-partial difference set} (PDS) if its internal difference list satisfies
\[
    D(B) = (\lambda \cdot B^*) \cup (\mu \cdot (G^* \setminus B)),
\]
where $B^* = B \setminus \{0\}$.
\end{definition}

The following is a well-known family of PDSs, called \emph{Paley-type}, which was first studied by Paley~\cite{paley1933orthogonal} in connection with Hadamard matrices.

\begin{theorem}[\cite{paley1933orthogonal}]
\label{thm_payley}
Let $q \equiv 1 \pmod{4}$ be a prime power. Then the subgroup $C_0^2$ of quadratic residues modulo $q$ 
forms a $(q, (q-1)/2, (q-5)/4, (q-1)/4)$-PDS.
\end{theorem}

We now turn to external difference families, which provide the main combinatorial framework for the AMD code constructions considered in this paper.

\begin{definition}[External Difference Family]
\label{def_EDF}
Let $G$ be an additive Abelian group of order $n$. A set system $\D = \{B_i : 1 \le i \le m\}$ of disjoint subsets of $G$ is called an \emph{$(n, K, \lambda)$-external difference family} (EDF) if its external difference list satisfies
\begin{equation}
\label{eq_EDF}
\bigcup_{1 \le i \neq j \le m} D(B_i, B_j) = \lambda \cdot G^*,
\end{equation}
where $K = \{|B_1|, |B_2|, \dots, |B_m|\}$ is the multiset of block sizes. Following the notation in Section II, we write $K=k_1^{u_1}\cdots k_t^{u_t}$
when convenient. 
An EDF is called \emph{regular} if
$|B_1|=\cdots=|B_m|=k$,
in which case it is denoted by an
$(n,k^m,\lambda)$-EDF.
\end{definition}

If an $(n,k^m,\lambda)$-EDF exists, then
\[
\lambda(n-1)=k^2m(m-1).
\]
Moreover, since the blocks are pairwise disjoint, we necessarily have $n\ge mk$.

Paterson and Stinson~\cite{paterson2016combinatorial}
generalized the notion of EDF as follows.

\begin{definition}[Bounded Generalized External Difference Family]
\label{def_BGEDF}
Let $G$ be an additive Abelian group of order $n$, and let
$\D=\{B_i:1\le i\le m\}$ be a set system of disjoint subsets
of $G$.
Let $K=\{|B_1|,|B_2|,\dots,|B_m|\}$ be the multiset of block
sizes.
If the external difference list of $\D$ satisfies
\begin{equation}
\label{eq_bgedf}
\bigcup_{1\le i\neq j\le m} D(B_i,B_j)
\subseteq
\lambda\cdot G^*,
\end{equation}
then $\D$ is called an
\emph{$(n,m,K,\lambda)$-bounded generalized external
difference family} (BGEDF).
If equality holds in~\eqref{eq_bgedf}, then $\D$ is called an
\emph{$(n,m,K,\lambda)$-generalized external difference
family} (GEDF).
In the uniform case where
$|B_1|=\cdots=|B_m|=k$,
a BGEDF is called an
\emph{$(n,k^m,\lambda)$-bounded external difference family}
(BEDF). If equality holds in~\eqref{eq_bgedf}, then it reduces
to an $(n,k^m,\lambda)$-EDF.
\end{definition}

As a special case of an $(n, k^m, \lambda)$-BEDF, Huczynska and Johnson \cite{huczynska2023internal} introduced the notion of an \emph{external partial difference family}.

\begin{definition}[External Partial Difference Family]
\label{def_EPDF}
Let $G$ be an additive Abelian group of order $n$, and let $\D = \{B_i : 1 \le i \le m\}$ be a set system of disjoint $k$-subsets of $G^*$. Define $H = \bigcup_{i=1}^m B_i$. 
If the external difference list of $\D$ satisfies 
\begin{equation}
\label{eqn_Delta_D}
\bigcup_{1 \le i \neq j \le m} D(B_i, B_j) = (\lambda \cdot H) \cup (\mu \cdot (G^* \setminus H)),
\end{equation}
then $\D$ is called an $(n, k^m, \lambda, \mu)$-\emph{external partial difference family} (EPDF).
\end{definition}

The following generalization of EDF was also introduced by
Paterson and Stinson~\cite{paterson2016combinatorial}.

\begin{definition}[Generalized Strong External Difference Family]
\label{def_GSEDF}
Let $G$ be an additive Abelian group of order $n$, and let
$\D=\{B_i:1\le i\le m\}$ be a set system of disjoint subsets
of $G$ with
$K=\{|B_i|:1\le i\le m\}$.
If
\begin{equation}
\label{eqn_GSEDF}
\bigcup_{j\in[m]\setminus\{i\}} D(B_i,B_j)
=
\lambda_i\cdot G^*
\end{equation}
holds for all $1\le i\le m$, then $\D$ is called an
$(n,K,\Lambda)$-\emph{generalized strong external difference
family} (GSEDF), where
$\Lambda=\{\lambda_i:1\le i\le m\}$.
\end{definition}

A bounded version of GSEDF is defined similarly~\cite{paterson2016combinatorial}.

\begin{definition}[Bounded Generalized Strong External Difference Family]
\label{def_BGSEDF}
Let $G$ be an additive Abelian group of order $n$, and let
$\D=\{B_i:1\le i\le m\}$ be a set system of disjoint subsets
of $G$ with
$K=\{|B_i|:1\le i\le m\}$.
If
\begin{equation}
\label{eqn_BGSEDF}
\bigcup_{j\in[m]\setminus\{i\}} D(B_i,B_j)
\subseteq
\lambda_i\cdot G^*
\end{equation}
holds for all $1\le i\le m$, then $\D$ is called an
$(n,K,\Lambda)$-\emph{bounded generalized strong external
difference family} (BGSEDF), where
$\Lambda=\{\lambda_i:1\le i\le m\}$.
In the uniform case where
$|B_1|=\cdots=|B_m|=k$,
it is denoted by an
$(n,k^m,\lambda)$-BGSEDF, where
$\lambda=\max\{\lambda_i:1\le i\le m\}$.
\end{definition}

The combinatorial structures introduced in this section form
the basis for the AMD code constructions developed in the
sequel. Section~IV focuses on the weak security model via
EPDFs and EDFs, while Section~V studies the strong security
model through (B)GSEDFs.

\section{External partial difference families and R-optimal weak AMD codes}
\label{miwakoEPDF}

In this section, we use the cyclotomic tools introduced in Section~III to construct external partial difference families (EPDFs) and the corresponding weak AMD codes.
We further investigate conditions under which these EPDFs become external difference families (EDFs).
Applying these criteria, we obtain new families of
$R$-optimal weak AMD codes attaining the bound in
Theorem~\ref{thm_weak_R_bound} with equality.

The following is a method to construct weak AMD codes from external partial difference families. 

\begin{construction}\label{con_weak_AMD}
Let $G$ be an additive Abelian group of order $n$.
Let $\A=\{A_1,\ldots,A_m\}$ be an
$(n,k^m,\lambda,\mu)$-EPDF, where
$A_i=\{a_{i,1},\ldots,a_{i,k}\}\subseteq G^*$.
Define an equiprobable encoding function $\Enc$
from the equiprobable source set
$S=\{s_1,\ldots,s_m\}$ to $G$ by
\begin{equation}\label{eqn_en_f_cons2}
\Enc(s_i)\triangleq a_{i,j}\in A_i,
\end{equation}
where $j\in_R[k]$.
The decoding function $\Dec:G\to S\cup\{\bot\}$
is defined by
\begin{equation}\label{eqn_de_f_cons2}
\Dec(g)\triangleq
\begin{cases}
s_i & \text{if } g\in A_i,\\
\bot & \text{otherwise}.
\end{cases}
\end{equation}
\end{construction}

The following theorem analyzes the weak security parameter of the AMD code obtained from Construction~\ref{con_weak_AMD}.

\begin{theorem}\label{lem_weak_AMD}
The AMD code generated by Construction~\ref{con_weak_AMD}
is a $k$-regular weak
$(m,n,k^m,\rhoEwk)$-AMD code, where
\[
\rhoEwk
=
\max\left\{
\frac{\lambda}{mk},
\frac{\mu}{mk}
\right\}.
\]
\end{theorem}

\begin{IEEEproof}
For a fixed $\Delta\in G^*$, the success probability in the
weak security model is
\begin{equation}\label{eqn_rho_delta}
\begin{split}
\rhoEE{\sigma}
&=
\sum_{s\in S}\Pr[s]
\sum_{g\in A_s}\Pr[\Enc(s)=g]
\sum_{s'\in S\setminus\{s\}}
\Pr[g+\Delta\in A_{s'}\mid s\in S,\Delta\in G^*]
\\
&=
\sum_{s\in S}\frac{1}{mk}
\sum_{s'\in S\setminus\{s\}}
\sum_{g\in A_s}
\Pr[g+\Delta\in A_{s'}\mid s\in S,\Delta\in G^*]
\\
&=
\begin{cases}
\dfrac{\lambda}{mk}
& \text{if }\Delta\in \bigcup_{i=1}^m A_i,\\[3mm]
\dfrac{\mu}{mk}
& \text{if }\Delta\in G^*\setminus \bigcup_{i=1}^m A_i. 
\end{cases}
\end{split}
\end{equation}
The last equality follows from~\eqref{eqn_Delta_D}.
Since the sources and encodings are equiprobable and
$|A_i|=k$ for all $1\le i\le m$, the resulting AMD code
is $k$-regular.
\end{IEEEproof}

\medskip
By Definitions~\ref{def_EDF} and~\ref{def_EPDF},
if $\lambda=\mu$, then an
$(n,k^m,\lambda,\mu)$-EPDF reduces to an
$(n,k^m,\lambda)$-EDF.
In this case,
\[
\lambda=\frac{k^2m(m-1)}{n-1},
\]
and, since the blocks are pairwise disjoint,
necessarily $n\ge km$.
Consequently, when an EPDF becomes an EDF,
Theorem~\ref{lem_weak_AMD} yields the following characterization
of $R$-optimal weak AMD codes.

\begin{theorem}[\cite{paterson2016combinatorial}]
\label{thm_weak_AMD}
An $(n,k^m,\lambda)$-EDF is equivalent to an
$R$-optimal $k$-regular weak
$(m,n,k^m,k(m-1)/(n-1))$-AMD code.
\end{theorem}

In what follows, we construct EPDFs using cyclotomy over
finite fields.
An element $\alpha\in\bF_q$ is called a
\emph{primitive $t$-th root of unity}
if its multiplicative order is $t$.
Equivalently,
$\alpha^t=1$ and $\alpha^i\ne1$ for
$1\le i\le t-1$.
If $\gamma$ is a primitive element of $\bF_q$, then the
primitive $t$-th roots of unity are precisely the elements
$\gamma^{h(q-1)/t}$,
where $1\le h\le t-1$ and $\gcd(h,t)=1$. 
Hence there are exactly $\varphi(t)$ primitive
$t$-th roots of unity in $\bF_q$, where
$\varphi$ denotes Euler's totient function.

\begin{lemma}
\label{lem_kawaguchi4.2}
For $k\in\bN$ and 
an odd prime power $q$, let
$\varepsilon_{2k}$ be a primitive $2k$-th root of unity
in $\bF_q$. Then
\[
\prod_{j=1}^{k-1}(\varepsilon_{2k}^j-1)^2
=
(-1)^{k-1}k\varepsilon_{2k}^l,
\]
where
$l\equiv \frac{k(k-1)}{2}\pmod{2k}$.
\end{lemma}

\begin{IEEEproof}
Since the roots of $x^{2k}-1=0$ are
$1,\varepsilon_{2k},\ldots,\varepsilon_{2k}^{2k-1}$,
we have
\[
x^{2k}-1
=
(x-1)\prod_{j=1}^{2k-1}(x-\varepsilon_{2k}^j).
\]
Comparing this with
\[
x^{2k}-1
=
(x-1)(x^{2k-1}+x^{2k-2}+\cdots+x+1)
\]
gives
\[
\prod_{j=1}^{2k-1}(x-\varepsilon_{2k}^j)
=
x^{2k-1}+x^{2k-2}+\cdots+x+1.
\]
Substituting $x=1$, we obtain
\begin{equation}
\label{eq_kawaguchi4.1}
\prod_{j=1}^{2k-1}(1-\varepsilon_{2k}^j)=2k.
\end{equation}
Since $\varepsilon_{2k}^k=-1$ and, for $1\le j\le k-1$,
\[
\varepsilon_{2k}^j-1
=
\varepsilon_{2k}^j(1-\varepsilon_{2k}^{2k-j}),
\]
we obtain
\begin{equation}
\label{eq_kawaguchi4.3}
\begin{split}
\prod_{j=1}^{k-1}(\varepsilon_{2k}^j-1)
&=
\prod_{j=1}^{k-1}
\varepsilon_{2k}^j(1-\varepsilon_{2k}^{2k-j})  \\
&=
\varepsilon_{2k}^{\frac{k(k-1)}{2}}
\prod_{j=k+1}^{2k-1}(1-\varepsilon_{2k}^j).
\end{split}
\end{equation}

Combining~\eqref{eq_kawaguchi4.1}
and~\eqref{eq_kawaguchi4.3}, we obtain
\begin{align*}
\prod_{j=1}^{k-1}(\varepsilon_{2k}^j-1)^2
&=
(-1)^{k-1}
\varepsilon_{2k}^{\frac{k(k-1)}{2}}
(1-\varepsilon_{2k}^k)^{-1}
\prod_{j=1}^{2k-1}(1-\varepsilon_{2k}^j)
\\
&=
(-1)^{k-1}
\varepsilon_{2k}^{\frac{k(k-1)}{2}}
\cdot \frac12 \cdot 2k
\\
&=
(-1)^{k-1}
k\varepsilon_{2k}^{\frac{k(k-1)}{2}},
\end{align*}
as desired.
\end{IEEEproof}

\begin{lemma}
\label{lem_kawaguchi4.3}
Let $q=2ef+1$ be a prime power such that $ef$ is even, and let
$\varepsilon_e$ be a primitive $e$-th root of unity in $\bF_q$.
Then, for each integer $1\le j<e/2$, the elements
$\varepsilon_e^j-1$ and $\varepsilon_e^{e-j}-1$ belong to the
same quadratic residue class in $\bF_q^*$.
\end{lemma}

\begin{IEEEproof}
Since $\varepsilon_e^e=1$ and $\varepsilon_e\ne1$, we have
\[
\sum_{i=0}^{e-1}\varepsilon_e^i
=
\frac{\varepsilon_e^e-1}{\varepsilon_e-1}
=0.
\]
Hence, for $1\le j<e/2$,
\begin{align*}
\varepsilon_e^{e-j}-1
&=(\varepsilon_e-1)\sum_{i=0}^{e-j-1}\varepsilon_e^i \\
&=-(\varepsilon_e-1)\sum_{i=e-j}^{e-1}\varepsilon_e^i \\
&=-(\varepsilon_e-1)\varepsilon_e^{e-j}
\sum_{i=0}^{j-1}\varepsilon_e^i \\
&=-\varepsilon_e^{e-j}(\varepsilon_e^j-1).
\end{align*}
Let $\gamma$ be a primitive element of $\bF_q$.
Then
$\varepsilon_e=\gamma^{(q-1)h/e}
=\gamma^{2fh}\in C_0^2$
for some $1\le h\le e-1$ with $\gcd(h,e)=1$.
Moreover,
$-1=\gamma^{ef}\in C_0^2$
because $ef$ is even.
Hence
$-\varepsilon_e^{e-j}\in C_0^2$.
The final identity above therefore shows that
$\varepsilon_e^{e-j}-1$ and $\varepsilon_e^j-1$
differ by multiplication by a quadratic residue.
Hence they belong to the same quadratic residue class
in $\bF_q^*$.
\end{IEEEproof}

\medskip
We now present a cyclotomic construction related to a known
construction of difference systems of sets (DSSs)
due to Mutoh and Tonchev~\cite[Theorem~2]{mutoh2008difference}.

\begin{theorem}
\label{thm_kurimoto4.4}
Let $p=4km+1$ be a prime with $m\ge2$, and let
$\gamma$ be a fixed primitive element of $\bF_p$.
Define
\[
S_{2k}
=
\{\varepsilon_{2k}^j-1:1\le j\le k-1\},
\]
where
$\varepsilon_{2k}=\gamma^{2m}$
is a primitive $2k$-th root of unity in $\bF_p$.
For $0\le i\le m-1$, set
$A_i:=C_{2i}^{2m}$,
and let
$\A=\{A_i:0\le i\le m-1\}$.
Then $\A$ is a
$(p,(2k)^m,k(m-2)+2\alpha+d,km-(2\alpha+d))$-EPDF,
where $\alpha=|S_{2k}\cap C_1^2|$,
and
\[
d=
\begin{cases}
0 & \text{if } p\equiv1\pmod8
\ \ (\text{i.e., } km \text{ is even}), \\
1 & \text{if } p\equiv5\pmod8
\ \ (\text{i.e., } km \text{ is odd}).
\end{cases}
\]
\end{theorem}

\begin{IEEEproof}
We first compute the internal difference list of
$C_{2i}^{2m}$:
\begin{align*}
D(C_{2i}^{2m})
&= \gamma^{2i}D(C_0^{2m}) \\
&= \gamma^{2i}
\{\varepsilon_{2k}^s-\varepsilon_{2k}^t:
0\le s\ne t\le 2k-1\} \\
&= \gamma^{2i}
\{\varepsilon_{2k}^t(\varepsilon_{2k}^{s-t}-1):
0\le s\ne t\le 2k-1\} \\
&= C_{2i}^{2m}
\{\varepsilon_{2k}^j-1:
1\le j\le 2k-1\}.
\end{align*}
Hence,
\begin{equation}
\label{eq_St}
\bigcup_{i=0}^{m-1} D(C_{2i}^{2m})
=
C_0^2
\{\varepsilon_{2k}^j-1:
1\le j\le 2k-1\}.
\end{equation}

By Theorem~\ref{thm_payley} and~\eqref{eq_St},
\begin{align*}
\bigcup_{0\le i\ne j\le m-1}
D(C_{2i}^{2m},C_{2j}^{2m})
&=
\left(
\left(
\frac{p-5}{4}
-
|\tilde S_{2k}\cap C_0^2|
\right)
\cdot C_0^2
\right)
\\
&\quad \cup
\left(
\left(
\frac{p-1}{4}
-
|\tilde S_{2k}\cap C_1^2|
\right)
\cdot C_1^2
\right),
\end{align*}
where
$\tilde S_{2k}
=
\{\varepsilon_{2k}^j-1:
1\le j\le 2k-1\}$.
Using
\[
|\tilde S_{2k}\cap C_0^2|
+
|\tilde S_{2k}\cap C_1^2|
=
2k-1,
\]
this becomes
\begin{align*}
\bigcup_{0\le i\ne j\le m-1}
D(C_{2i}^{2m},C_{2j}^{2m})
&=
\left(
(k(m-2)+|\tilde S_{2k}\cap C_1^2|)
\cdot C_0^2
\right)
\\
&\quad \cup
\left(
(km-|\tilde S_{2k}\cap C_1^2|)
\cdot C_1^2
\right).
\end{align*}

Now $\varepsilon_{2k}^k=-1$, so
$\varepsilon_{2k}^k-1=-2$.
Define
\[
d=
\begin{cases}
1 & \text{if } -2\in C_1^2,\\
0 & \text{if } -2\in C_0^2.
\end{cases}
\]
By Lemma~\ref{lem_kawaguchi4.3},
\[
|\tilde S_{2k}\cap C_1^2|
=
2|S_{2k}\cap C_1^2|+d
=
2\alpha+d.
\]
We have
$p=4km+1\equiv1\pmod4$,
so the first supplement to quadratic reciprocity gives
\[
\left(\frac{-1}{p}\right)=1.
\]
Moreover, by the second supplement to quadratic reciprocity, 
\[
\left(\frac{2}{p}\right)=
\begin{cases}
1 & \text{if } p\equiv1\pmod8
\ \ (\text{i.e., } km \text{ is even}), \\
-1 & \text{if } p\equiv5\pmod8
\ \ (\text{i.e., } km \text{ is odd}).
\end{cases}
\]
Thus,
\[
\left(\frac{-2}{p}\right)
=
\left(\frac{-1}{p}\right)
\left(\frac{2}{p}\right)
=
\left(\frac{2}{p}\right),
\]
which determines $d$ as stated.
\end{IEEEproof}

\medskip
Based on Theorem~\ref{lem_weak_AMD}
and Theorem~\ref{thm_kurimoto4.4},
we obtain the following family of weak AMD codes.

\begin{corollary}
Let $p=4km+1$ be a prime with $m\ge2$.
Then there exists a weak AMD code with parameters
\[
\left(
m,
p,
(2k)^m,
\frac{\max\{k(m-2)+2\alpha+d,\ km-(2\alpha+d)\}}
{2km}
\right),
\]
where $\alpha$ and $d$ are as in Theorem~\ref{thm_kurimoto4.4}.
\end{corollary}

Note that the EPDF in Theorem~\ref{thm_kurimoto4.4}
becomes an EDF if and only if
\[
k(m-2)+2\alpha+d
=
km-(2\alpha+d),
\]
which is equivalent to
\begin{equation}
\label{eq_gamma}
\alpha
=
|S_{2k}\cap C_1^2|
=
\frac{k-d}{2}.
\end{equation}
Hence, the EPDF in Theorem~\ref{thm_kurimoto4.4}
is an EDF precisely in the following cases.

\begin{theorem}
\label{thm_EDF}
Let $p=4km+1$ be a prime with $m\ge2$.
The EPDF in Theorem~\ref{thm_kurimoto4.4}
is a $(p,(2k)^m,k(m-1))$-EDF if and only if one of the following holds:
\begin{enumerate}
\item[(1)]
$km$ is odd and
$\alpha=|S_{2k}\cap C_1^2|
=|S_{2k}\cap C_0^2|
=(k-1)/2$;

\item[(2)]
$k$ is even and
$\alpha=|S_{2k}\cap C_1^2|=k/2$
(equivalently,
$|S_{2k}\cap C_0^2|=k/2-1$).
\end{enumerate}
\end{theorem}

Theorem~\ref{thm_EDF} implicitly states that when $m$ is even and $k$ is odd,
any EPDF from Theorem~\ref{thm_kurimoto4.4} can never be an EDF.

Although Theorem~\ref{thm_EDF} gives a complete characterization,
computing the exact value of
$\alpha=|S_{2k}\cap C_1^2|$
quickly becomes expensive as $k$ grows.
By considering only the parity of $\alpha$,
we obtain the following simpler necessary condition.

\begin{corollary}
\label{cor_EDF}
Let $p=4km+1$ be a prime with $m\ge2$.
If the EPDF in Theorem~\ref{thm_kurimoto4.4}
is an EDF, then
\[
k^{\frac{p-1}{4}}
\equiv
\begin{cases}
1 \pmod p
& \text{if $m$ is odd or $k\equiv0\pmod4$,} \\
-1 \pmod p
& \text{if $m$ is even and $k\equiv2\pmod4$.}
\end{cases}
\]
\end{corollary}

\begin{IEEEproof}
By Lemma~\ref{lem_kawaguchi4.2}, we have
\[
\prod_{x \in S_{2k}} x^2 =
(-1)^{k-1} k\varepsilon_{2k}^l \in C_{2i}^4,
\]
where $l \equiv \frac{k(k-1)}{2} \pmod{2k}$, and $i = 0$ or $1$
depending on whether
$\alpha = |S_{2k} \cap C_1^2|$
is even or odd.
Since
$(-1)^{k-1}
=\gamma^{2k(k-1)m}$
and $2k(k-1)m\equiv0\pmod4$,
we have
$(-1)^{k-1}\in C_0^4$.
Hence
\begin{equation}
\label{eq_simple-cond}
k\varepsilon_{2k}^l \in C_{2i}^4.
\end{equation}
It follows from \eqref{eq_gamma} that for $\alpha = \frac{k - d}{2}$ to be an integer, $k$ and $d$ must have the same parity.
The parity of $\alpha$ then depends on whether $d = 0$ or $d = 1$, i.e., whether $km$ is even or odd. 
If $d = 1$, then $\alpha = \frac{k-1}{2}$, 
which is even when $k \equiv 1 \pmod{4}$ and odd when $k \equiv 3 \pmod{4}$. 
If $d = 0$, then $\alpha = \frac{k}{2}$, 
which is even when $k \equiv 0 \pmod{4}$ and odd when $k \equiv 2 \pmod{4}$.
Theorem~\ref{thm_EDF} implies that in \eqref{eq_simple-cond}, we have
\[
  i = \begin{cases}
  0, & \mbox{if $k \equiv 0 \pmod{4}$, or $m$ is odd and $k \equiv 1 \pmod{4}$,} \\
  1, & \mbox{if $k \equiv 2 \pmod{4}$, or $m$ is odd and $k \equiv 3 \pmod{4}$,}
  \end{cases}
\]
under the assumption that the EPDF from Theorem~\ref{thm_kurimoto4.4} is an EDF.

Since $\varepsilon_{2k} = \gamma^{2m}$, we have
\begin{equation}
  \label{eq_simple-cond2}
  \varepsilon_{2k}^l = \gamma^{mk(k-1)} \in C_{2i}^4,
\end{equation}
with
\[
  i = \begin{cases}
    0, & \mbox{if $m$ is even or $k \equiv 0, 1 \pmod{4}$,} \\
    1, & \mbox{if $m$ is odd and $k \equiv 2,3 \pmod{4}$.}
  \end{cases}
\]
The values of $i$ in \eqref{eq_simple-cond} and~\eqref{eq_simple-cond2} are summarized in Table~\ref{tab_simple-cond}.
\begin{table}[htbp]
 \begin{center}
   \caption{Values of $i$ in (\ref{eq_simple-cond})
   and (\ref{eq_simple-cond2});
   entries marked N/A correspond to cases excluded by Theorem~\ref{thm_EDF}. 
    \label{tab_simple-cond}}
\begin{tabular}{cc||cc|cc|cc|cc}\hline
   & $k \pmod{4}$ & \multicolumn{2}{c|}{$0$} & \multicolumn{2}{c|}{$1$} & \multicolumn{2}{c|}{$2$} & \multicolumn{2}{c}{$3$} \\
   & & $k\varepsilon_{2k}^l$ & $\varepsilon_{2k}^l$ & $k\varepsilon_{2k}^l$ & $\varepsilon_{2k}^l$ & $k\varepsilon_{2k}^l$ & $\varepsilon_{2k}^l$ & $k\varepsilon_{2k}^l$ & $\varepsilon_{2k}^l$ \\ 
 \hline
   \multirow{2}{*}{$m$} & odd & 0 & 0 & 0 & 0 & 1 & 1 & 1 & 1 \\
   & even & 0 & 0 & {\footnotesize N/A} & 0 & 1 & 0 & {\footnotesize N/A} & 0 \\ \hline
\end{tabular}
\end{center}
\end{table}
From Table~\ref{tab_simple-cond}, we conclude that $k \in C_{2i}^4$,
and hence $k^{(p-1)/4} \equiv (-1)^i \pmod{p}$, where
\[
  i = \begin{cases}
  0 & \text{if $m$ is odd or $k \equiv 0 \pmod{4}$,}\\
  1 & \text{if $m$ is even and $k \equiv 2 \pmod{4}$,}
  \end{cases}
\]
which completes the proof.
\end{IEEEproof}

\medskip
By \cite[Note~2 and Theorem~12]{mutoh2008difference},
the necessary condition in Corollary~\ref{cor_EDF}
is also sufficient when $k=2$ or $3$
in the setting of perfect regular DSSs.
In the remainder of this section, we examine the cases
$k=4,5,6,$ and $7$ individually and derive explicit
criteria for the resulting EPDFs to become EDFs.
For larger values of $k$, the corresponding residue conditions
become increasingly complicated.
To this end, we first prove the following lemma.

\begin{lemma}
  \label{lem_e}
For $m \ge 2$, let $p=4km+1$ be a prime and $\varepsilon_{2k}$ be a primitive $2k$-th root of unity in $\bF_p$.
Then, for $1 \le i < k$, we have
\[
\left(\frac{\varepsilon_{2k}^{k-i}-1}{p}\right) = \left(\frac{\varepsilon_{2k}^i+1}{p}\right).
\]
\end{lemma}

\begin{IEEEproof}
Let $\gamma$ be a primitive element of $\bF_p$.
Since
$\varepsilon_{2k}=\gamma^{2m}\in C_0^2$
and
$-1=\varepsilon_{2k}^k$,
the assertion follows from
\[
\varepsilon_{2k}^{k-i}-1
=
\varepsilon_{2k}^{k-i}
+
\varepsilon_{2k}^k
=
\varepsilon_{2k}^{k-i}
(\varepsilon_{2k}^i+1).
\]
\end{IEEEproof}

\begin{corollary}
  \label{cor_k4}
  Let $p = 16m + 1$ be a prime with $m \ge 2$. Then $\D = \{C_{2i}^{2m}: \ 0 \le i \le m-1\}$ forms a $(p,8^m,4(m-1))$-EDF if and only if
  at least one of $\left(\tfrac{\varepsilon_8 - 1}{p}\right)$ and $\left(\tfrac{\varepsilon_8 + 1}{p}\right)$ equals $-1$,
  where $\varepsilon_8$ denotes a primitive $8$-th root of unity in $\bF_p$.
\end{corollary}

\begin{IEEEproof}
By Theorem~\ref{thm_EDF} with $k = 4$, the set system $\D$ is a $(p, 8^m, 4(m - 1))$-EDF
if and only if Condition~(2) of Theorem~\ref{thm_EDF} holds, that is, $|S_8 \cap C_1^2| = 2$ (equivalently, $|S_8 \cap C_0^2| = 1$).

Note that
\begin{align*}
  \varepsilon_8^3 - 1 &= \varepsilon_8^3 + \varepsilon_8^4 = \varepsilon_8^3(\varepsilon_8 + 1), \\
  \varepsilon_8^2 - 1 &= (\varepsilon_8 - 1)(\varepsilon_8 + 1),
\end{align*}
and since $\left(\tfrac{\varepsilon_8}{p}\right) = \left(\tfrac{\gamma^{2m}}{p}\right) = 1$, we have $\varepsilon_8 \in C_0^2$.
Therefore, $|S_8 \cap C_1^2| = 2$ if and only if at least one of $\varepsilon_8 - 1$ and $\varepsilon_8 + 1$ lies in $C_1^2$,
which is equivalent to
\[
\left(\frac{\varepsilon_8 - 1}{p}\right) = -1 \quad \text{or} \quad \left(\frac{\varepsilon_8 + 1}{p}\right) = -1.
\]
\end{IEEEproof}

\medskip
By Theorem~\ref{thm_weak_AMD}
and Corollary~\ref{cor_k4},
we obtain the following family of optimal weak AMD codes.

\begin{corollary}
Let $p=16m+1$ be a prime with $m\ge2$.
If at least one of
$\left(\tfrac{\varepsilon_8-1}{p}\right)$
and
$\left(\tfrac{\varepsilon_8+1}{p}\right)$
equals $-1$,
where $\varepsilon_8$ denotes a primitive
$8$-th root of unity in $\bF_p$,
then there exists an R-optimal
$8$-regular weak
$(m,p,8^m,\frac{m-1}{2m})$-AMD code.
\end{corollary}

\begin{example}\rm
  \label{ex_k4}
  The following table lists the values of
  $\left(\tfrac{\varepsilon_8-1}{p}\right)$
  and
  $\left(\tfrac{\varepsilon_8+1}{p}\right)$
  for primes of the form $p=16m+1$ with $p<700$.
  For each prime $p$, a primitive element $\gamma$ of $\bF_p$ was chosen and
  $\varepsilon_8=\gamma^{2m}$ was computed accordingly.
  Rows shaded in gray indicate primes that do not satisfy the condition given in Corollary~\ref{cor_k4}.
  {\small
  \begin{center}
  \begin{tabular}{r|r||r|r||r|r}\hline
    \multicolumn{1}{c|}{$p$} & $m$ & \multicolumn{1}{c|}{$\gamma$} & \multicolumn{1}{c||}{$\varepsilon_8$} & $\left(\tfrac{\varepsilon_8-1}{p}\right)$ & $\left(\tfrac{\varepsilon_8+1}{p}\right)$ \\ \hline
    $97$ & $6$ & $5$ & $64$ & $-1$ & $1$ \\
    $113$ & $7$ & $3$ & $18$ & $-1$ & $-1$ \\
    $193$ & $12$ & $5$ & $43$ & $1$ & $-1$ \\
    $241$ & $15$ & $7$ & $30$ & $1$ & $-1$ \\
    $257$ & $16$ & $3$ & $64$ & $-1$ & $-1$ \\
    \rowcolor{black!10} $337$ & $21$ & $10$ & $85$ & $1$ & $1$ \\
    $353$ & $22$ & $3$ & $283$ & $-1$ & $-1$ \\
    $401$ & $25$ & $3$ & $45$ & $1$ & $-1$ \\
    $433$ & $27$ & $5$ & $354$ & $-1$ & $1$ \\
    $449$ & $28$ & $3$ & $122$ & $1$ & $-1$ \\
    \rowcolor{black!10} $577$ & $36$ & $5$ & $391$ & $1$ & $1$ \\
    \rowcolor{black!10} $593$ & $37$ & $3$ & $392$ & $1$ & $1$ \\
    $641$ & $40$ & $3$ & $318$ & $1$ & $-1$ \\
    $673$ & $42$ & $5$ & $609$ & $-1$ & $1$ \\\hline
  \end{tabular}
  \end{center}
  }
\end{example}

\begin{corollary}
\label{cor_k5}
Let $p=20m+1$ be a prime with $m\ge2$. Then
$\D=\{C_{2i}^{2m}:0\le i\le m-1\}$ forms a
$(p,10^m,5(m-1))$-EDF if and only if $m$ is odd and
\begin{equation}
\label{eq_k5}
\left(\frac{\varepsilon_{10}-1}{p}\right)=1
\quad\text{and}\quad
\left(\frac{\varepsilon_{10}+1}{p}\right)=-1,
\end{equation}
where $\varepsilon_{10}$ denotes a primitive $10$-th root of unity in
$\bF_p$.
\end{corollary}

\begin{IEEEproof}
By Theorem~\ref{thm_EDF} with $k=5$, the set system $\D$ is a
$(p,10^m,5(m-1))$-EDF if and only if Condition~(1) of
Theorem~\ref{thm_EDF} holds, namely, $m$ is odd and
$|S_{10}\cap C_1^2|=|S_{10}\cap C_0^2|=2$.

As in the proof of Lemma~\ref{lem_e}, we use
\begin{align*}
\varepsilon_{10}^4-1
&=(\varepsilon_{10}^2+1)(\varepsilon_{10}^2-1),\\
\varepsilon_{10}^3-1
&=\varepsilon_{10}^3(\varepsilon_{10}^2+1).
\end{align*}
Together with
$\varepsilon_{10}^2-1=(\varepsilon_{10}-1)(\varepsilon_{10}+1)$,
this gives
\[
\prod_{x\in S_{10}}x
=
\prod_{i=1}^4(\varepsilon_{10}^i-1)
=
\varepsilon_{10}^3
(\varepsilon_{10}-1)(\varepsilon_{10}^2-1)^2
(\varepsilon_{10}^2+1)^2.
\]
Since $\varepsilon_{10}\in C_0^2$, the condition
$|S_{10}\cap C_1^2|=2$ implies
$\varepsilon_{10}-1\in C_0^2$.
Under this assumption, the remaining possible combinations
of Legendre symbols are summarized in the following table:
\begin{center}
{\small
\begin{tabular}{c||c|c|c}\hline
\multirow{2}{*}{$\varepsilon_{10}-1$}
&
$\varepsilon_{10}+1$
&
$\varepsilon_{10}^2+1$
&
\multirow{2}{*}{$\varepsilon_{10}^2-1$}
\\
&
$(\varepsilon_{10}^4-1)$
&
$(\varepsilon_{10}^3-1)$
&
\\ \hline\hline
\multirow{2}{*}{\textbf{1}}
& $1$ & $1$ & $1$ \\
& $-1$ & $1$ & $-1$ \\ \hline
\end{tabular}
}
\end{center}
Here we also used Lemma~\ref{lem_e}.
The second row of the table is therefore equivalent to \eqref{eq_k5},
which proves the corollary.
\end{IEEEproof}

\medskip
By Theorem~\ref{thm_weak_AMD}
and Corollary~\ref{cor_k5},
we obtain the following family of optimal weak AMD codes.

\begin{corollary}
Let $p=20m+1$ be a prime with $m\ge2$.
If $m$ is odd and \eqref{eq_k5} holds,
then there exists an R-optimal
$10$-regular weak
$(m,p,10^m,\frac{m-1}{2m})$-AMD code.
\end{corollary}

\begin{example}\rm
  \label{ex_k5}
  The following table lists the values of
  $\left(\tfrac{\varepsilon_{10}-1}{p}\right)$
  and
  $\left(\tfrac{\varepsilon_{10}+1}{p}\right)$
  for primes of the form
  $p=20m+1$ with odd $m$ and $p<1500$.
  For each prime $p$, a primitive element $\gamma$ of $\bF_p$
  was chosen and
  $\varepsilon_{10}=\gamma^{2m}$ was computed accordingly.
  Rows shaded in gray indicate primes that do not satisfy the condition in Corollary~\ref{cor_k5}.
  \begin{center}
  {\small
  \begin{tabular}{r|r||r|r||r|r}\hline
    \multicolumn{1}{c|}{$p$} & $m$ & \multicolumn{1}{c|}{$\gamma$} & \multicolumn{1}{c||}{$\varepsilon_{10}$} & $\left(\tfrac{\varepsilon_{10}-1}{p}\right)$ & $\left(\tfrac{\varepsilon_{10}+1}{p}\right)$ \\ \hline
    \rowcolor{black!10} $61$ & $3$ & $2$ & $3$ & $-1$ & $1$ \\
    $101$ & $5$ & $2$ & $14$ & $1$ & $-1$ \\
    $181$ & $9$ & $2$ & $56$ & $1$ & $-1$ \\
    \rowcolor{black!10} $421$ & $21$ & $2$ & $67$ & $-1$ & $1$ \\
    \rowcolor{black!10} $461$ & $23$ & $2$ & $347$ & $1$ & $1$ \\
    \rowcolor{black!10} $541$ & $27$ & $2$ & $313$ & $1$ & $1$ \\
    \rowcolor{black!10} $661$ & $33$ & $2$ & $255$ & $-1$ & $1$ \\
    \rowcolor{black!10} $701$ & $35$ & $2$ & $612$ & $-1$ & $1$ \\
    \rowcolor{black!10} $821$ & $41$ & $2$ & $660$ & $-1$ & $-1$ \\
    $941$ & $47$ & $2$ & $592$ & $1$ & $-1$ \\
    \rowcolor{black!10} $1021$ & $51$ & $10$ & $432$ & $1$ & $1$ \\
    \rowcolor{black!10} $1061$ & $53$ & $2$ & $406$ & $1$ & $1$ \\
    \rowcolor{black!10} $1181$ & $59$ & $7$ & $1100$ & $-1$ & $1$ \\
    \rowcolor{black!10} $1301$ & $65$ & $2$ & $282$ & $-1$ & $1$ \\
    \rowcolor{black!10} $1381$ & $69$ & $2$ & $1306$ & $-1$ & $-1$ \\ \hline
  \end{tabular}
  }
  \end{center}
  The first ten primes satisfying
  Corollary~\ref{cor_k5} are
  \[
    101,\ 181,\ 941,\ 1861,\ 3301,\ 3581,\ 4021,\ 4621,\ 4861,\ 5101.
  \]
\end{example}

\begin{corollary}
  \label{cor_k6}
  Let $p = 24m +1$ be a prime with $m \ge 2$.
  Then $\D=\{C_{2i}^{2m}: \ 0\le i \le m-1\}$ forms a $(p,12^m,6(m-1))$-EDF if and only if
  \begin{equation}
  \label{eq_k6}
  \left(\frac{(\varepsilon_{12}^2+1)(\varepsilon_{12}^3-1)}{p}\right) = \left(\frac{\varepsilon_{12}+1}{p}\right) = \left(\frac{\varepsilon_{12}-1}{p}\right) = -1,
  \end{equation}
  where $\varepsilon_{12}$ denotes a primitive $12$-th root of unity in $\bF_p$.
\end{corollary}

\begin{IEEEproof}
By Theorem~\ref{thm_EDF} with $k = 6$, the set $\D$ is a $(p,12^m,6(m-1))$-EDF
if and only if Condition~(2) of Theorem~\ref{thm_EDF} holds; that is, $|S_{12} \cap C_1^2| = 3$ (equivalently, $|S_{12} \cap C_0^2| = 2$).

We use the factorizations
\begin{align*}
  & \varepsilon_{12}^5 - 1 = \varepsilon_{12}^5 + \varepsilon_{12}^6 = \varepsilon_{12}^5(\varepsilon_{12}+1), \\
  & \varepsilon_{12}^4 - 1 = \varepsilon_{12}^4 + \varepsilon_{12}^6 = \varepsilon_{12}^4(\varepsilon_{12}^2+1)
  = (\varepsilon_{12}^2 + 1)(\varepsilon_{12}^2 - 1), \\
  & \varepsilon_{12}^2 - 1 = (\varepsilon_{12}+1)(\varepsilon_{12}-1).
\end{align*}
Together with Lemma~\ref{lem_e}, these identities give
\[
\prod_{x \in S_{12}} x
=
\prod_{i=1}^5 (\varepsilon_{12}^i - 1)
=
\varepsilon_{12}^9
(\varepsilon_{12}-1)^2
(\varepsilon_{12}+1)^2
(\varepsilon_{12}^2+1)
(\varepsilon_{12}^3-1).
\]
Hence
$\prod_{x\in S_{12}}x\in C_1^2$,
since $|S_{12}\cap C_1^2|=3$ is odd. 
Since $\varepsilon_{12}\in C_0^2$,
this is equivalent to
$(\varepsilon_{12}^2+1)(\varepsilon_{12}^3-1)\in C_1^2$.
Under this assumption, the remaining possible combinations
of Legendre symbols are summarized in the following table:
\begin{center}
{\small
\begin{tabular}{c|c||c|c|c}\hline
$\varepsilon_{12}^2+1$
& \multirow{2}{*}{$\varepsilon_{12}^3-1$}
& \multirow{2}{*}{$\varepsilon_{12}^2-1$}
& $\varepsilon_{12}+1$
& \multirow{2}{*}{$\varepsilon_{12}-1$} \\
$(\varepsilon_{12}^4-1)$
&
&
& $(\varepsilon_{12}^5-1)$
& \\ \hline\hline
\multirow{2}{*}{$\phantom{-}\bm{1}$}
& \multirow{2}{*}{$\bm{-1}$}
& \multirow{4}{*}{$1$}
& $\phantom{-}1$ & $\phantom{-}1$ \\
& & & $-1$ & $-1$ \\ \cline{1-2}
\multirow{2}{*}{$\bm{-1}$}
& \multirow{2}{*}{$\phantom{-}\bm{1}$}
&
& $\phantom{-}1$ & $\phantom{-}1$ \\
& & & $-1$ & $-1$ \\ \hline
\end{tabular}
}
\end{center}
The row in which both
$\varepsilon_{12}+1$ and $\varepsilon_{12}-1$
are quadratic nonresidues is therefore equivalent to
\eqref{eq_k6},
which proves the corollary.
\end{IEEEproof}

\medskip
By Theorem~\ref{thm_weak_AMD}
and Corollary~\ref{cor_k6},
we obtain the following family of optimal weak AMD codes.

\begin{corollary}
Let $p=24m+1$ be a prime with $m\ge2$.
If \eqref{eq_k6} holds, then there exists an
R-optimal $12$-regular weak
$(m,p,12^m,\frac{m-1}{2m})$-AMD code.
\end{corollary}

\begin{example}\rm
\label{ex_k6}
The following table lists the values of
$\left(\tfrac{(\varepsilon_{12}^2+1)(\varepsilon_{12}^3-1)}{p}\right)$,
$\left(\tfrac{\varepsilon_{12}+1}{p}\right)$,
and
$\left(\tfrac{\varepsilon_{12}-1}{p}\right)$
for primes of the form
$p=24m+1$
with $p<1000$.
For each prime $p$, a primitive element $\gamma$ of $\bF_p$
was chosen and
$\varepsilon_{12}=\gamma^{2m}$ was computed accordingly.
Rows shaded in gray indicate primes that do not satisfy
the condition in Corollary~\ref{cor_k6}.
\begin{center}
{\small
\begin{tabular}{r|r||r|r||r|r|r}\hline
\multicolumn{1}{c|}{$p$}
& $m$
& \multicolumn{1}{c|}{$\gamma$}
& \multicolumn{1}{c||}{$\varepsilon_{12}$}
&
$\left(\tfrac{(\varepsilon_{12}^2+1)(\varepsilon_{12}^3-1)}{p}\right)$
&
$\left(\tfrac{\varepsilon_{12}+1}{p}\right)$
&
$\left(\tfrac{\varepsilon_{12}-1}{p}\right)$
\\ \hline
\rowcolor{black!10} $73$ & $3$ & $5$ & $3$ & $1$ & $1$ & $1$ \\
\rowcolor{black!10} $97$ & $4$ & $5$ & $6$ & $1$ & $-1$ & $-1$ \\
\rowcolor{black!10} $193$ & $8$ & $5$ & $49$ & $-1$ & $1$ & $1$ \\
\rowcolor{black!10} $241$ & $10$ & $7$ & $237$ & $1$ & $1$ & $1$ \\
\rowcolor{black!10} $313$ & $13$ & $10$ & $284$ & $1$ & $-1$ & $-1$ \\
$337$ & $14$ & $10$ & $265$ & $-1$ & $-1$ & $-1$ \\
\rowcolor{black!10} $409$ & $17$ & $21$ & $360$ & $-1$ & $1$ & $1$ \\
$433$ & $18$ & $5$ & $64$ & $-1$ & $-1$ & $-1$ \\
\rowcolor{black!10} $457$ & $19$ & $13$ & $18$ & $-1$ & $1$ & $1$ \\
\rowcolor{black!10} $577$ & $24$ & $5$ & $57$ & $1$ & $-1$ & $-1$ \\
\rowcolor{black!10} $601$ & $25$ & $7$ & $481$ & $-1$ & $1$ & $1$ \\
$673$ & $28$ & $5$ & $16$ & $-1$ & $-1$ & $-1$ \\
\rowcolor{black!10} $769$ & $32$ & $11$ & $750$ & $-1$ & $1$ & $1$ \\
$937$ & $39$ & $5$ & $408$ & $-1$ & $-1$ & $-1$ \\\hline
\end{tabular}
}
\end{center}
The first ten primes $p=24m+1$
satisfying Corollary~\ref{cor_k6} are
\[
337,\ 433,\ 673,\ 937,\ 1129,\ 1249,\ 2113,\ 2617,\ 3001,\ 3313.
\]
\end{example}

As shown in Corollaries~\ref{cor_k4}--\ref{cor_k6},
the necessary condition in Corollary~\ref{cor_EDF}
is generally not sufficient for an EPDF to become an EDF.
The case $k=7$, however, is exceptional:
the condition in Theorem~\ref{thm_EDF}
is both necessary and sufficient.

\begin{corollary}
  \label{cor_k7}
  Let $p = 28m + 1$ be a prime with odd $m\ge3$. 
  Then $\D = \{C_{2i}^{2m}: \ 0 \leq i \leq m - 1\}$ forms a $(p,14^m,7(m-1))$-EDF if and only if $7$ is a biquadratic residue modulo $p$.
\end{corollary}

\begin{IEEEproof}
By Theorem~\ref{thm_EDF} with $k = 7$, $\D$ forms a $(p,14^m,7(m-1))$-EDF if and only if Condition~(1) of Theorem~\ref{thm_EDF} holds; that is, 
$|S_{14} \cap C_1^2| = |S_{14} \cap C_0^2| = 3$.

We use the factorizations
\begin{align*}
  \varepsilon_{14}^6 - 1 &= \varepsilon_{14}^6 + \varepsilon_{14}^7 = \varepsilon_{14}^6(\varepsilon_{14} + 1) = (\varepsilon_{14}^3 + 1)(\varepsilon_{14}^3 - 1), \\
  \varepsilon_{14}^5 - 1 &= \varepsilon_{14}^5 + \varepsilon_{14}^7 = \varepsilon_{14}^5(\varepsilon_{14}^2 + 1), \\
  \varepsilon_{14}^4 - 1 &= \varepsilon_{14}^4 + \varepsilon_{14}^7 = \varepsilon_{14}^4(\varepsilon_{14}^3 + 1) = (\varepsilon_{14}^2 + 1)(\varepsilon_{14}^2 - 1), \\
  \varepsilon_{14}^2 - 1 &= (\varepsilon_{14} + 1)(\varepsilon_{14} - 1).
\end{align*}
Together with Lemma~\ref{lem_e}, these identities give
\[
\prod_{x \in S_{14}} x
=
\prod_{i=1}^6 (\varepsilon_{14}^i - 1)
=
\varepsilon_{14}^9
(\varepsilon_{14}-1)^2
(\varepsilon_{14}+1)
(\varepsilon_{14}^2+1)
(\varepsilon_{14}^3-1)^2
(\varepsilon_{14}^3+1)^2.
\]
Hence
$\prod_{x\in S_{14}}x\in C_1^2$,
since $|S_{14}\cap C_1^2|=3$ is odd.
As $\varepsilon_{14}\in C_0^2$,
this is equivalent to
$(\varepsilon_{14}+1)(\varepsilon_{14}^2+1)\in C_1^2$.

Under this assumption, the remaining possible combinations
of Legendre symbols are summarized in the following table:
\begin{center}
{\small
\begin{tabular}{c|c||c|c|c|c}\hline
$\varepsilon_{14}+1$
& $\varepsilon_{14}^2+1$
& $\varepsilon_{14}^3+1$
& \multirow{2}{*}{$\varepsilon_{14}^3-1$}
& \multirow{2}{*}{$\varepsilon_{14}^2-1$}
& \multirow{2}{*}{$\varepsilon_{14}-1$} \\
$(\varepsilon_{14}^6-1)$
& $(\varepsilon_{14}^5-1)$
& $(\varepsilon_{14}^4-1)$
&
&
& \\ \hline\hline
\multirow{2}{*}{$\phantom{-}\bm{1}$}
& \multirow{2}{*}{$\bm{-1}$}
& $\phantom{-}1$
& $\phantom{-}1$
& $-1$
& $-1$ \\
&
&
$-1$
& $-1$
& $\phantom{-}1$
& $\phantom{-}1$ \\ \hline
\multirow{2}{*}{$\bm{-1}$}
& \multirow{2}{*}{$\phantom{-}\bm{1}$}
& $\phantom{-}1$
& $-1$
& $\phantom{-}1$
& $-1$ \\
&
&
$-1$
& $\phantom{-}1$
& $-1$
& $\phantom{-}1$ \\ \hline
\end{tabular}
}
\end{center}
From the table, we see that
$(\varepsilon_{14}+1)(\varepsilon_{14}^2+1)\in C_1^2$
implies
$|S_{14}\cap C_1^2|=3$.
Hence, for $k=7$, the parity condition
$|S_{14}\cap C_1^2|\equiv1\pmod2$
is equivalent to Condition~(1) of
Theorem~\ref{thm_EDF}.
By Corollary~\ref{cor_EDF}, this is equivalent to
$7^{(p-1)/4}\equiv1\pmod p$,
that is, to $7\in C_0^4$.
\end{IEEEproof}

\medskip
By Theorem~\ref{thm_weak_AMD}
and Corollary~\ref{cor_k7},
we obtain the following family of optimal weak AMD codes.

\begin{corollary}
  Let $p = 28m + 1$ be a prime with odd $m \geq 3$.
  Then there exists an R-optimal $14$-regular weak $(m,p,14^m,\frac{m-1}{2m})$-AMD code 
  if $7$ is a biquadratic residue modulo $p$.
\end{corollary}

\begin{example}\rm
\label{ex_k7}
The following table lists the values of
$\left(\tfrac{(\varepsilon_{14}+1)(\varepsilon_{14}^2+1)}{p}\right)$
for primes of the form
$p=28m+1$
with $p<2500$.
For each prime $p$, a primitive element $\gamma$ of $\bF_p$
was chosen and
$\varepsilon_{14}=\gamma^{2m}$ was computed accordingly.
Rows shaded in gray indicate primes for which the corresponding value is $1$, and hence the condition in Corollary~\ref{cor_k7} fails.
\begin{center}
{\small
\begin{tabular}{r|r||r|r||r}\hline
\multicolumn{1}{c|}{$p$}
& $m$
& \multicolumn{1}{c|}{$\gamma$}
& \multicolumn{1}{c||}{$\varepsilon_{14}$}
&
$\left(\tfrac{(\varepsilon_{14}+1)(\varepsilon_{14}^2+1)}{p}\right)$
\\ \hline
\rowcolor{black!10} $197$ & $7$ & $2$ & $33$ & $1$ \\
\rowcolor{black!10} $421$ & $15$ & $2$ & $269$ & $1$ \\
$701$ & $25$ & $2$ & $65$ & $-1$ \\
$757$ & $27$ & $2$ & $127$ & $-1$ \\
$1093$ & $39$ & $5$ & $1090$ & $-1$ \\
$1373$ & $49$ & $2$ & $745$ & $-1$ \\
$1429$ & $51$ & $6$ & $1227$ & $-1$ \\
\rowcolor{black!10} $1597$ & $57$ & $11$ & $1278$ & $1$ \\
\rowcolor{black!10} $1709$ & $61$ & $3$ & $1541$ & $1$ \\
\rowcolor{black!10} $1877$ & $67$ & $2$ & $1212$ & $1$ \\
\rowcolor{black!10} $1933$ & $69$ & $5$ & $1154$ & $1$ \\
$2213$ & $79$ & $2$ & $1873$ & $-1$ \\
$2269$ & $81$ & $2$ & $2020$ & $-1$ \\
\rowcolor{black!10} $2381$ & $85$ & $3$ & $1360$ & $1$ \\
\rowcolor{black!10} $2437$ & $87$ & $2$ & $1945$ & $1$ \\ \hline
\end{tabular}
}
\end{center}
The first ten primes $p=28m+1$
satisfying Corollary~\ref{cor_k7} are
\[
701,\ 757,\ 1093,\ 1373,\ 1429,\ 2213,\ 2269,\ 3109,\ 3389,\ 4229.
\]
\end{example}

\section{Systematic strong AMD codes via bounded generalized strong external difference families}\label{sec_cons_cosets}

So far, we have focused on weak AMD codes constructed from
EDF-type structures.
We now turn to systematic strong AMD codes, which provide stronger protection against algebraic manipulation attacks.

In general, systematic AMD codes offer the additional advantage that 
the source value appears explicitly as a component of the encoded word. 
This feature is particularly useful in scenarios where the privacy of $s\in S$ is not a concern  
and only its authenticity must be protected, 
since $s$ can be recovered
without applying the decoding function.

By Definition \ref{def_BGSEDF}, an $(n,\{k_1,\ldots,k_m\},1)$-BGSEDF is simply a collection of disjoint subsets of an additive Abelian group $G$ of order $n$, 
each of size $k_i$, $1 \le i \le m$, such that each element of $G^*$ occurs at most once 
as a difference between elements in different subsets.
Such structures are closely related to G-optimal strong AMD codes.
In particular, it was shown in \cite{paterson2016combinatorial}
that BGSEDFs characterize G-optimal strong AMD codes as follows.

\begin{theorem}[\cite{paterson2016combinatorial}]
An $(n,\{k_1,\ldots,k_m\},1)$-BGSEDF
is equivalent to a G-optimal strong AMD code
with parameters
$(m,n,\{k_1,\ldots,k_m\},\rho)$
under equiprobable encoding,
where
$\rho=\max_{1\le i\le m}\frac1{k_i}$.
\end{theorem}

Throughout this section,
let
$G=S\times G_1\times G_2$
be a finite additive Abelian group,
where
$S$, $G_1$, and $G_2$
are written additively.

The goal of this section is to investigate systematic G-optimal strong AMD codes via BGSEDFs.
We begin with a simple observation,
which follows from the definitions of G-optimal systematic AMD codes and BGSEDFs.


\begin{definition}[Systematic BGSEDF]
\label{def_systematic}
Let $\A=\{A_s:s\in S\}$ be an $(n,|G_1|^m,1)$-BGSEDF over an additive Abelian group $G=S\times G_1\times G_2$. The family $\A$ is said to be systematic if, for some function $f:S\times G_1\to G_2$, $A_s=\{(s,x,f(s,x)):x\in G_1\}$ for every $s\in S$.
\end{definition}

\begin{proposition}
\label{prop_sysAMD_sysBGSEDF}
A systematic G-optimal strong
$(m,n,|G_1|,1/|G_1|)$-AMD code
$(S,S\times G_1\times G_2,\A,\Enc)$
is equivalent to a systematic
$(n,|G_1|^m,1)$-BGSEDF
$\A=\{A_s:s\in S\}$.
\end{proposition}

\begin{IEEEproof}
Suppose first that
$(S,S\times G_1\times G_2,\A,\Enc)$
is a systematic G-optimal strong
$(m,n,|G_1|,1/|G_1|)$-AMD code, 
where
\[
\Enc(s)=(s,x,f(s,x)),
\quad x\in_R G_1.
\]
Let
\[
\A=\{A_s:s\in S\},
\qquad
A_s=\{(s,x,f(s,x)):x\in G_1\}.
\]
It follows that
$A_s\subseteq \{s\}\times G_1\times G_2$
and
$|A_s|=|G_1|$
for every $s\in S$.
We claim that $\A$ is an
$(n,|G_1|^m,1)$-BGSEDF.
Otherwise,  by Definition~\ref{def_BGSEDF}, 
there would exist
$s\in S$, a nonzero difference
$\Delta\in S\times G_1\times G_2$ 
that occurs at least twice in
\[
\bigcup_{s'\in S\setminus\{s\}} D(A_s,A_{s'}).
\]
Equivalently, there would exist distinct elements
$g,g'\in A_s$
such that
$g+\Delta\in A_{s'}$
and
$g'+\Delta\in A_{s''}$
for some
$s',s''\in S\setminus\{s\}$.
Since the encoding is equiprobable over $A_s$,
this would give a successful tampering probability at least
$2/|A_s|=2/|G_1|$ for the source $s$,
contradicting G-optimality.
Thus $\A$ is an
$(n,|G_1|^m,1)$-BGSEDF.

Conversely, suppose that
$\A=\{A_s:s\in S\}$ 
is a systematic $(n,|G_1|^m,1)$-BGSEDF. 
By Definition~\ref{def_systematic},
there exists a function
$f:S\times G_1\to G_2$
such that
$A_s=\{(s,x,f(s,x)):x\in G_1\}$
for every $s\in S$.
Define the systematic encoding by choosing
$x\in_R G_1$
and setting
\[
\Enc(s)=(s,x,f(s,x)),
\]
and define the decoding function by
\[
\Dec(s,x,t)=
\begin{cases}
s & \text{if }  t=f(s,x),\\
\bot & \text{otherwise}.
\end{cases}
\]
Since $\A$ is a BGSEDF with $\lambda=1$,
for every nonzero $\Delta\in S\times G_1\times G_2$
and every source $s\in S$, there is at most one element
$g\in A_s$ such that
$g+\Delta\in A_{s'}$ for some $s'\ne s$.
Thus the success probability is at most
$1/|A_s|=1/|G_1|$.
On the other hand, by Theorem~\ref{thm_sys_bound},
every systematic strong AMD code satisfies
$\rho_E^{\mathrm{str}}\ge 1/|G_1|$.
Hence the resulting systematic strong AMD code attains this bound
and is therefore G-optimal.
\end{IEEEproof}

\medskip
By Proposition~\ref{prop_sysAMD_sysBGSEDF}, the construction of systematic
G-optimal strong AMD codes is equivalent to the construction
of systematic BGSEDFs. 
In the remainder of this section,
we present three constructions of systematic BGSEDFs:
the first is based on cyclotomy over finite fields,
the second on direct products of finite fields,
and the third on generalized cyclotomy over residue class rings.
Each construction yields an infinite family of systematic
G-optimal strong AMD codes.

\subsection{Systematic BGSEDFs via Cyclotomic Classes over Finite Fields}

Cyclotomic classes have long been used in the construction of
difference families and external difference families (see,
for example, \cite{wilson1972cyclotomy}).
However, their connection with strong external difference families
appears to be much less explored.

In \cite{paterson2024circular}, cyclotomic classes over $\bF_q$
were used to construct AMD codes by taking the cyclotomic classes
themselves as image sets of the encoding function.
The following theorem shows that cyclotomic classes can also be
used to construct systematic BGSEDFs, which in turn yield
systematic G-optimal strong AMD codes via Proposition~\ref{prop_sysAMD_sysBGSEDF}.

\begin{theorem}
\label{thm_sysBGSEDF}
Let $q=eu+1$ be a prime power,
and let $\gamma$ be a primitive element of $\bF_q$.
For each $s\in\bZ_u$, define
\[
A_s=
\{(s,t,\gamma^{tu+s}):t\in\bZ_e\}.
\]
Then $\A=\{A_s:s\in\bZ_u\}$ forms a systematic
$(q(q-1),e^u,1)$-BGSEDF in $\bZ_u\times \bZ_e\times\bF_q$.
\end{theorem}

\begin{IEEEproof}
It is clear from the definition of the blocks that
the systematic condition in Definition~\ref{def_systematic}
is satisfied.
By Definition~\ref{def_BGSEDF},
it suffices to show that every difference between distinct blocks
occurs at most once.

Let $s_1,s_2\in\bZ_u$ with $s_1\ne s_2$.
Then
\[
D(A_{s_1},A_{s_2})
=
\{
(s_1-s_2,t_1-t_2,
\gamma^{t_1u+s_1}-\gamma^{t_2u+s_2})
:
t_1,t_2\in\bZ_e
\}.
\]
We first show that all elements of
$D(A_{s_1},A_{s_2})$
are distinct.
Suppose
\[
(t_1-t_2,\gamma^{t_1u+s_1}-\gamma^{t_2u+s_2})
=
(t'_1-t'_2,\gamma^{t'_1u+s_1}-\gamma^{t'_2u+s_2}).
\]
Then
\[
\gamma^{t_1u+s_1}
\bigl(1-\gamma^{(t_2-t_1)u+(s_2-s_1)}\bigr)
=
\gamma^{t'_1u+s_1}
\bigl(1-\gamma^{(t_2-t_1)u+(s_2-s_1)}\bigr).
\]
Since $s_1\ne s_2$ in $\mathbb Z_u$, we have
$(t_2-t_1)u+(s_2-s_1)\not\equiv0\pmod{eu}$, and hence
$1-\gamma^{(t_2-t_1)u+(s_2-s_1)}\ne0$.
Cancelling this common factor gives
$\gamma^{t_1u+s_1}=\gamma^{t'_1u+s_1}$, so
$\gamma^{(t_1-t'_1)u}=1$.
Since $\gamma$ is a primitive element of $\mathbb F_q$, it has order
$q-1=eu$. Thus $t_1=t'_1$ in $\mathbb Z_e$, and the equality
$t_1-t_2=t'_1-t'_2$ then yields $t_2=t'_2$.
Therefore, all elements of $D(A_{s_1},A_{s_2})$ are distinct.

It remains to show that
\[
D(A_{s_1},A_{s_2})
\cap
D(A_{s_1},A_{s_3})
=
\emptyset
\]
for distinct
$s_2,s_3\in\bZ_u\setminus\{s_1\}$.
This follows immediately because the first coordinates
$s_1-s_2$ and $s_1-s_3$ are different.
Therefore every difference between distinct blocks
occurs at most once.
This completes the proof.
\end{IEEEproof}

\medskip
By Proposition~\ref{prop_sysAMD_sysBGSEDF},
the systematic $(q(q-1),e^u,1)$-BGSEDF
constructed in Theorem~\ref{thm_sysBGSEDF} 
gives rise to the following family of
systematic G-optimal strong AMD codes.

\begin{corollary}\label{cor_cyclotomic}
Let $q=eu+1$ be a prime power.
Then there exists a systematic G-optimal strong
$(u,q(q-1),e^u,1/e)$-AMD code.
\end{corollary}

\begin{IEEEproof}
The result follows immediately from
Proposition~\ref{prop_sysAMD_sysBGSEDF}
and Theorem~\ref{thm_sysBGSEDF}.
\end{IEEEproof}

\begin{example}
Let $q=25$ and $u=6$.
Construct $\bF_{25}$ using the primitive polynomial
$x^2+4x+2\in\bF_5[x]$.
Then the cyclotomic classes of index $6$ are
\begin{equation*}
\begin{array}{ll}
C_0^6=\{1,2,4,3\},&
C_1^6=\{x,2x,4x,3x\}, \\
C_2^6=\{x+3,2x+1,4x+2,3x+4\},&
C_3^6=\{4x+3,3x+1,x+2,2x+4\},\\
C_4^6=\{2x+2,4x+4,3x+3,x+1\},&
C_5^6=\{4x+1,3x+2,x+4,2x+3\}.
\end{array}
\end{equation*}
Applying Theorem~\ref{thm_sysBGSEDF},
we obtain a systematic
$(600,4^6,1)$-BGSEDF.
By Corollary~\ref{cor_cyclotomic},
this yields a systematic G-optimal strong $(6,600,4^6,\frac{1}{4})$-AMD code.
\end{example}

The finite-field construction presented above can be extended to direct products of finite fields, which will be considered in the next subsection.

\subsection{Systematic BGSEDFs via Cyclotomy over Direct Products of Finite Fields}

Throughout this subsection, let
$q_1,q_2,\ldots,q_\ell$
be distinct prime powers and let
$e>1$
be an integer satisfying
$e\mid\gcd(q_1-1,q_2-1,\ldots,q_\ell-1)$.
Set
\[
X=
\bigl(
\bF_{q_1}\times\cdots\times\bF_{q_\ell}
\bigr)
\setminus\{(0,\ldots,0)\}.
\]
For
$\bm a=(a_1,\ldots,a_\ell)$
and
$\bm b=(b_1,\ldots,b_\ell)$
in $X$, define
$\bm a\bm b=(a_1b_1,\ldots,a_\ell b_\ell)$.
For
$\bm a\in X$
and a subset
$B\subseteq X$,
write
$\bm aB=\{\bm a\bm b:\bm b\in B\}$.
Choose elements
$\delta_i\in\bF_{q_i}^*$
of multiplicative order $e$
for $1\le i\le \ell$, and put
$\bm{\delta}=(\delta_1,\delta_2,\ldots,\delta_\ell)$.
For $k\ge0$, define
\[
\bm{\delta}^k
=
(\delta_1^k,\delta_2^k,\ldots,\delta_\ell^k).
\]
Let
$H=\langle\bm{\delta}\rangle
=\{\bm{\delta}^k:0\le k\le e-1\}$.
Define a relation on $X$ by
$\bm a\sim\bm b$
if and only if
$\bm a,\bm b\in \bm dH$
for some
$\bm d\in X$.
Then $\sim$ is an equivalence relation.

\begin{theorem}\label{thm:Fq1qtBGSEDF}
Let
$v=\prod_{i=1}^{\ell}q_i$,
and let
$L=\{\alpha_0,\alpha_1,\ldots,\alpha_{(v-1)/e-1}\}$ 
be a set of representatives for the equivalence classes of $\sim$.
For each $s\in\bZ_{(v-1)/e}$, define
\[
A_s=
\{
(s,\tau,\alpha_s\bm{\delta}^{\tau})
:
\tau\in\bZ_e
\},
\]
where $\alpha_s$ denotes the $s$th element of $L$.
Then
$\A=\{A_s:s\in\bZ_{(v-1)/e}\}$
forms a systematic
$(v(v-1),e^{(v-1)/e},1)$-BGSEDF in
$\bZ_{(v-1)/e}\times\bZ_e
\times
(\bF_{q_1}\times\cdots\times\bF_{q_\ell})$.
\end{theorem}

\begin{IEEEproof}
By Definition~\ref{def_systematic}, the family $\A$ is systematic.
We therefore focus on verifying the BGSEDF condition with $\lambda=1$.

Let $s_1\in\bZ_{(v-1)/e}$.
For $s_2\in\bZ_{(v-1)/e}\setminus\{s_1\}$, we have
\[
D(A_{s_1},A_{s_2})
=
\left\{
\left(
s_1-s_2,
\tau_1-\tau_2,
\alpha_{s_1}\bm{\delta}^{\tau_1}
-
\alpha_{s_2}\bm{\delta}^{\tau_2}
\right)
:
\tau_1,\tau_2\in\bZ_e
\right\}.
\]
We first show that all elements of
$D(A_{s_1},A_{s_2})$
are distinct.
Suppose that two elements corresponding to
$(\tau_1,\tau_2)$
and
$(\tau'_1,\tau'_2)$
coincide, where
$(\tau_1,\tau_2)\ne(\tau'_1,\tau'_2)$.
If
$\tau_1-\tau_2\ne\tau'_1-\tau'_2$,
then the two elements are distinct.
Hence we may assume that
$\tau_1-\tau_2=\tau'_1-\tau'_2$.
Then
$\tau_2-\tau_1=\tau'_2-\tau'_1$
and
$\tau_1\ne\tau'_1$.
Comparing the third coordinates gives
\[
\alpha_{s_1}\bm{\delta}^{\tau_1}
-
\alpha_{s_2}\bm{\delta}^{\tau_2}
=
\alpha_{s_1}\bm{\delta}^{\tau'_1}
-
\alpha_{s_2}\bm{\delta}^{\tau'_2}.
\]
Using
$\tau_2-\tau_1=\tau'_2-\tau'_1$,
this equality becomes
\[
\alpha_{s_1}
\left(
\bm{\delta}^{\tau_1}
-
\bm{\delta}^{\tau'_1}
\right)
=
\alpha_{s_2}\bm{\delta}^{\tau_2-\tau_1}
\left(
\bm{\delta}^{\tau_1}
-
\bm{\delta}^{\tau'_1}
\right).
\]
Because $\delta_i$ has order $e$ and
$0\le \tau_1\ne\tau'_1\le e-1$,
we have
$\delta_i^{\tau_1}\ne\delta_i^{\tau'_1}$
for every $i$.
Therefore
\[
\bm{\delta}^{\tau_1}
-
\bm{\delta}^{\tau'_1}
\in
\bF_{q_1}^*\times\cdots\times\bF_{q_\ell}^*.
\]
This element is invertible in
$\bF_{q_1}\times\cdots\times\bF_{q_\ell}$,
so cancellation yields
\[
\alpha_{s_1}
=
\alpha_{s_2}\bm{\delta}^{\tau_2-\tau_1}.
\]
Thus
$\alpha_{s_1}$ and $\alpha_{s_2}$
belong to the same equivalence class, contrary to the choice of $L$
and the assumption $s_1\ne s_2$.
Therefore
$D(A_{s_1},A_{s_2})$
has no repeated elements.

We next show that
$D(A_{s_1},A_{s_2})$
and
$D(A_{s_1},A_{s_2'})$
are disjoint for distinct
$s_2,s_2'\in\bZ_{(v-1)/e}\setminus\{s_1\}$.
The elements of
$D(A_{s_1},A_{s_2})$
have first coordinate $s_1-s_2$, whereas those of
$D(A_{s_1},A_{s_2'})$
have first coordinate $s_1-s_2'$.
Hence
\[
D(A_{s_1},A_{s_2})
\cap
D(A_{s_1},A_{s_2'})
=
\emptyset.
\]
Since $s_1$ was arbitrary, the BGSEDF condition holds with $\lambda=1$.
Hence $\A$ is a systematic
$(v(v-1),e^{(v-1)/e},1)$-BGSEDF.
\end{IEEEproof}

\begin{corollary}\label{cor:Fq1qtAMD}
Let
$v=\prod_{i=1}^{\ell}q_i$.
Then there exists a systematic G-optimal strong AMD code
with parameters
\[
\left(
\frac{v-1}{e},
v(v-1),
e^{(v-1)/e},
\frac1e
\right).
\]
\end{corollary}

\begin{IEEEproof}
The result follows immediately from
Theorem~\ref{thm:Fq1qtBGSEDF}
and Proposition~\ref{prop_sysAMD_sysBGSEDF}.
\end{IEEEproof}

\subsection{Systematic BGSEDFs via Generalized Cyclotomic Classes over Integer Residue Rings}

In the previous subsections, we constructed systematic BGSEDFs from cyclotomic classes over finite fields and their direct products. We next consider a related construction based on the generalized cyclotomy of Zeng et al.~\cite{zeng2013optimal} over residue class rings.

We briefly recall the required notation and results from \cite{zeng2013optimal}. These yield a partition of suitable subsets of $\bZ_v$, which will be used to construct systematic BGSEDFs.

For an odd positive integer $v$, let $v=p_1^{m_1}p_2^{m_2}\cdots p_\ell^{m_\ell}$, where
$p_1, p_2, \dots, p_\ell$ are $\ell$ primes with $2<p_1<p_2<\cdots<p_\ell$ and
$m_1, m_2, \dots, m_\ell$ are positive integers.
Let $\bZ_v^\times =
\{x\in\bZ_v:\gcd(x,v)=1\}$
denote the group of units of $\bZ_v$.
If $C_0$ is a multiplicative subgroup of $\bZ^\times_{v}$ and
$C_i=a_iC_0$ for some elements $a_i\in\bZ^\times_{v}$, $1\leq i\leq d-1$, 
then $C_0,C_1,\dots,C_{d-1}$ are called \emph{generalized cyclotomic classes} of index $d$ 
when $v$ is composite,
and \emph{classical cyclotomic classes} of index $d$ when $v$ is prime \cite{ding1998new}.

For any odd prime $p$, there exists a primitive root $\delta$ modulo $p$
that is also a primitive root modulo $p^i$ for all $i\ge 1$
\cite{apostol1998introduction}. 
For each $1\le i\le\ell$, fix such a primitive root
$\delta_i$ modulo $p_i^j$ for all $j\ge 1$.
Let $e>1$ be a common factor of
$p_1-1,\ldots,p_\ell-1$,
and write
$p_i-1=ef_i$
for $1\le i\le\ell$.
Consider a divisor
$d=p_{i_1}^{r_1}p_{i_2}^{r_2}\cdots p_{i_h}^{r_h}$
of $v$, where
$1\le i_1<i_2<\cdots<i_h\le\ell$
and
$0<r_j\le m_{i_j}$ for $1\le j\le h$.
Applying the Chinese Remainder Theorem,
there exists a unique $\alpha_d\in \bZ_d^\times$ such that
\begin{equation}\label{eqn_alpha_v}
\alpha_d
\equiv
\delta_{i_j}^{f_{i_j}p_{i_j}^{r_j-1}}
\pmod{p_{i_j}^{r_j}}
\qquad
\text{for }1\le j\le h.
\end{equation}
It is straightforward to check that the multiplicative order of
$\alpha_d$ modulo $d$ is $e$.
Hence
\[
C^{(d)}
=
\{\alpha_d^i:0\le i\le e-1\}
\]
is a cyclic subgroup of $\bZ_d^\times$ of order $e$.
For any
\[
I=(g_1,g_2,\ldots,g_h)
\in
\mathcal{I}(d)
\triangleq
\{0,1,\ldots,f_{i_1}p_{i_1}^{r_1-1}-1\}
\times
\bZ_{\varphi(p_{i_2}^{r_2})}
\times\cdots\times
\bZ_{\varphi(p_{i_h}^{r_h})},
\]
define $C_I^{(d)}$ by
\begin{equation}\label{eqn_def_DI_v}
C_{I}^{(d)}
=
\delta^IC^{(d)}
\triangleq
\delta_{i_1}^{g_1}
\delta_{i_2}^{g_2}
\cdots
\delta_{i_h}^{g_h}
C^{(d)}.
\end{equation}

In \cite{zeng2013optimal}, it was proved that
$\{C_{I}^{(d)}:I\in\mathcal{I}(d)\}$
forms a family of generalized cyclotomic classes of index
$\varphi(d)/e$ with respect to $d$.
Moreover, the collection
\begin{equation}\label{eq:psi-partition}
\Psi
\triangleq
\left\{
\frac{v}{d}C_I^{(d)}
:
d\mid v,\ d>1,\ I\in\mathcal{I}(d)
\right\}.
\end{equation}
forms a partition of $\bZ_v\setminus\{0\}$, and $|\Psi|=(v-1)/e$.

Let $\Theta:\bZ_{(v-1)/e}\to\Psi$ be a bijection.
For each
$s\in\bZ_{(v-1)/e}$,
define
\begin{equation}
\label{eq:Zv(v-1)_As}
A_s=
\left\{
\left(
s,\tau,
\frac{v}{d}\delta^I
\alpha_d^\tau
\right)
:
\tau\in\bZ_e
\right\},
\end{equation}
where
$\Theta(s)=\frac{v}{d}C_{I}^{(d)}$.

\begin{theorem}
\label{thm:cycBGSEDF}
Let
$\A=\{A_s:s\in\bZ_{(v-1)/e}\}$,
where the sets $A_s$ are defined by
\eqref{eq:Zv(v-1)_As}.
Then
$\A$
forms a systematic
$\bigl(v(v-1),e^{(v-1)/e},1\bigr)$-BGSEDF in
$\bZ_{(v-1)/e}\times\bZ_e\times\bZ_v$.
\end{theorem}

\begin{IEEEproof}
It is clear from the definition of the blocks~\eqref{eq:Zv(v-1)_As} that
$\A$ satisfies the systematic condition in
Definition~\ref{def_systematic}.
To prove that $\A$ is a
$\bigl(v(v-1),e^{(v-1)/e},1\bigr)$-BGSEDF,
it suffices to show that every external difference
occurs at most once.

For
$s_1,s_2\in\bZ_{(v-1)/e}$
with
$s_1\ne s_2$,
let
$\Theta(s_1)=\frac{v}{d}C_I^{(d)}$
and
$\Theta(s_2)=\frac{v}{d'}C_J^{(d')}$. 
Then 
\[
D(A_{s_1},A_{s_2})
=
\left\{
\left(
s_1-s_2,\tau_1-\tau_2,
\frac{v}{d}\delta^I\alpha_d^{\tau_1}
-
\frac{v}{d'}\delta^J\alpha_{d'}^{\tau_2}
\right)
:
\tau_1,\tau_2\in\bZ_e
\right\}.
\]
We next show that
$|D(A_{s_1},A_{s_2})|=e^2$
by considering the cases
$d\ne d'$ and $d=d'$ separately.

\emph{Case 1:}
Suppose that $d\ne d'$.
It suffices to show that all elements of
$D(A_{s_1},A_{s_2})$
are distinct.

Let
$(\tau_1,\tau_2),(\tau'_1,\tau'_2)\in\bZ_e^2$
with
$(\tau_1,\tau_2)\ne(\tau'_1,\tau'_2)$.
If
$\tau_1-\tau_2\ne\tau'_1-\tau'_2$,
then the elements of
$D(A_{s_1},A_{s_2})$
corresponding to
$(\tau_1,\tau_2)$
and
$(\tau'_1,\tau'_2)$
are distinct.
Hence, we may assume that
$\tau_1-\tau_2=\tau'_1-\tau'_2$.
Then there exists
$\Delta\in\bZ_e\setminus\{0\}$
such that
\[
(\tau'_1,\tau'_2)
=
(\tau_1+\Delta,\tau_2+\Delta).
\]
Assume, for a contradiction, that the elements of
$D(A_{s_1},A_{s_2})$
corresponding to
$(\tau_1,\tau_2)$
and
$(\tau'_1,\tau'_2)$
coincide.
Comparing the third coordinates, we obtain
\begin{equation}
\label{eqn_d_d'}
\frac{v}{d}\delta^I\alpha_d^{\tau_1}
-
\frac{v}{d'}\delta^J\alpha_{d'}^{\tau_2}
=
\frac{v}{d}\delta^I\alpha_d^{\tau_1+\Delta}
-
\frac{v}{d'}\delta^J\alpha_{d'}^{\tau_2+\Delta}.
\end{equation}
Write
\[
d=p_1^{r_1}p_2^{r_2}\cdots p_\ell^{r_\ell}
\quad\text{and}\quad
d'=p_1^{r'_1}p_2^{r'_2}\cdots p_\ell^{r'_\ell},
\]
where \(0\le r_j,r'_j\le m_j\).
Since \(d\ne d'\), there exists \(1\le i\le\ell\)
such that \(r_i\ne r'_i\).
Without loss of generality, assume that \(r_i>r'_i\).
Reducing \eqref{eqn_d_d'} modulo \(p_i^{m_i}\), we obtain
\[
\frac{v}{d'}\delta^J
\left(
\alpha_{d'}^{\tau_2+\Delta}
-\alpha_{d'}^{\tau_2}
\right)
\equiv
\frac{v}{d}\delta^I
\left(
\alpha_d^{\tau_1+\Delta}
-\alpha_d^{\tau_1}
\right)
\pmod{p_i^{m_i}}.
\]
Hence
\begin{equation}
\begin{split}
p_i^{r_i-r'_i}
& \prod_{\substack{1\le j\le\ell\\ j\ne i}}
p_j^{m_j-r'_j}\delta^J
\left(
\alpha_{d'}^{\tau_2+\Delta}
-\alpha_{d'}^{\tau_2}
\right)
\\
\equiv
& \prod_{\substack{1\le j\le\ell\\ j\ne i}}
p_j^{m_j-r_j}\delta^I
\left(
\alpha_d^{\tau_1+\Delta}
-\alpha_d^{\tau_1}
\right)
\pmod{p_i^{r_i}}
\\
\overset{(a)}{\Rightarrow}\quad&
0\equiv
\alpha_d^{\tau_1+\Delta}
-\alpha_d^{\tau_1}
\pmod{p_i}
\\
\overset{(b)}{\Rightarrow}\quad&
0\equiv
\delta_i^{(\tau_1+\Delta)f_i}
-\delta_i^{\tau_1f_i}
\pmod{p_i}.
\end{split}
\end{equation}
The implication marked $(a)$ follows from
$r_i>r'_i$ and
$\gcd(p_i,p_j)=1$ for $i\ne j$,
and the implication marked $(b)$ follows from
\eqref{eqn_alpha_v}.
The last congruence implies
$\delta_i^{\Delta f_i}\equiv1\pmod{p_i}$.
Since $\delta_i$ is a primitive root modulo $p_i$
and $p_i-1=ef_i$, we have $e\mid\Delta$,
contrary to
$\Delta\in\bZ_e\setminus\{0\}$.
Thus all elements of
$D(A_{s_1},A_{s_2})$
are distinct in this case.

\emph{Case 2:}
Suppose that $d=d'$.
Write
\[
\Theta(s_1)=\frac{v}{d}C_I^{(d)}
\quad\text{and}\quad
\Theta(s_2)=\frac{v}{d}C_J^{(d)}.
\]
Since $\Theta$ is injective and $s_1\ne s_2$,
we have
$C_I^{(d)}\ne C_J^{(d)}$.

As in Case 1, it suffices to show that all elements of
$D(A_{s_1},A_{s_2})$
are distinct.
Let
$(\tau_1,\tau_2),(\tau'_1,\tau'_2)\in\bZ_e^2$
with
$(\tau_1,\tau_2)\ne(\tau'_1,\tau'_2)$.
If
$\tau_1-\tau_2\ne\tau'_1-\tau'_2$,
then the elements of
$D(A_{s_1},A_{s_2})$ corresponding to 
$(\tau_1,\tau_2)$
and
$(\tau'_1,\tau'_2)$ 
are distinct.
Hence, as in Case 1, we may assume that
$(\tau'_1,\tau'_2)
=
(\tau_1+\Delta,\tau_2+\Delta)$
for some
$\Delta\in\bZ_e\setminus\{0\}$.

Assume, for a contradiction, that the elements of
$D(A_{s_1},A_{s_2})$ corresponding to
$(\tau_1,\tau_2)$
and
$(\tau'_1,\tau'_2)$
coincide.
Then, comparing the third coordinates, we have
\[
\frac{v}{d}\delta^I\alpha_d^{\tau_1}
-
\frac{v}{d}\delta^J\alpha_d^{\tau_2}
=
\frac{v}{d}\delta^I\alpha_d^{\tau_1+\Delta}
-
\frac{v}{d}\delta^J\alpha_d^{\tau_2+\Delta}.
\]
Equivalently, in $\bZ_{d}$,
\[
\alpha_d^{\tau_1}
\left(1-\alpha_d^\Delta\right)
\left(
\delta^I
-
\delta^J\alpha_d^{\tau_2-\tau_1}
\right)
=0.
\]
By \eqref{eqn_alpha_v},
the reduction of $\alpha_d$
modulo each prime divisor of $d$
has order $e$.
Since $\Delta\in\bZ_e\setminus\{0\}$,
we have $\alpha_d^\Delta\not\equiv1\pmod p$
for every prime divisor $p$ of $d$.
Hence
$1-\alpha_d^\Delta$
is a unit in $\bZ_{d}$.
Also, $\alpha_d^{\tau_1}$ is a unit.
It follows that
\[
\delta^I
=
\delta^J\alpha_d^{\tau_2-\tau_1}.
\]
Since $\alpha_d^{\tau_2-\tau_1}\in C^{(d)}$, this implies
$C_I^{(d)}=C_J^{(d)}$,
contrary to $C_I^{(d)}\ne C_J^{(d)}$.
Thus all elements of
$D(A_{s_1},A_{s_2})$
are distinct in this case.
Finally, if $s_2\ne s'_2$, then the first coordinates of the elements of
$D(A_{s_1},A_{s_2})$ and $D(A_{s_1},A_{s'_2})$ are
$s_1-s_2$ and $s_1-s'_2$, respectively.
Hence these two difference lists are disjoint.
Since $s_1$ was arbitrary, every external difference from a fixed block
to the other blocks occurs at most once.
Therefore the BGSEDF condition holds with $\lambda=1$.
\end{IEEEproof}

\begin{corollary}\label{cor:Gcyc}
Let
$v=p_1^{m_1}p_2^{m_2}\cdots p_\ell^{m_\ell}$,
where
$2<p_1<p_2<\cdots<p_\ell$
are primes and
$m_1,m_2,\ldots,m_\ell$
are positive integers.
Let $e>1$ be a common factor of
$p_1-1,p_2-1,\ldots,p_\ell-1$.
Then there exists a systematic G-optimal strong AMD code
with parameters
\[
\left(
\frac{v-1}{e},
\,v(v-1),
\,e^{(v-1)/e},
\,\frac1e
\right).
\]
\end{corollary}

\begin{IEEEproof}
The result follows immediately from
Proposition~\ref{prop_sysAMD_sysBGSEDF}
and Theorem~\ref{thm:cycBGSEDF}.
\end{IEEEproof}

\begin{example}
\label{ex:v65}
Let $v=65=5\cdot13$ and $e=4$.
Since $(v-1)/e=16$, Corollary~\ref{cor:Gcyc}
yields a systematic G-optimal strong
$(16,4160,4^{16},\tfrac14)$-AMD code.
\end{example}

Corollary~\ref{cor:Gcyc}
provides systematic G-optimal strong AMD codes for composite integers $v$
with at least two distinct prime factors,
thereby considerably extending the range of parameter sets obtainable from the finite-field constructions of the previous subsections.

\section{Concluding Remarks}\label{sec-conclusion}

In this paper, we constructed optimal algebraic manipulation detection
(AMD) codes from several classes of external difference families and
bounded generalized strong external difference families arising from
cyclotomic structures over finite fields, direct products of finite
fields, and residue class rings of integers.

\begin{enumerate}
\item Weak AMD codes: 

\ \ \ We gave a cyclotomic construction of EPDFs, leading to constructions of weak AMD codes.
We derived a necessary and sufficient criterion for these EPDFs to form EDFs
and demonstrated its effectiveness by obtaining explicit characterizations for small block sizes.
Our results yielded new families of EDFs as well as concrete examples verified
through quadratic and biquadratic residue computations.
As an immediate consequence, these results yield new R-optimal weak AMD codes. 

\ \ \ \ Looking forward, several research directions remain open. 
First, while our analysis provided explicit criteria for small values of $k$, 
extending these results to larger parameters ($k \ge 8$) poses significant challenges. 
The necessary and sufficient conditions appear to become increasingly complicated, 
and it is unlikely that simple residue-based characterizations, such as those obtained for $k=7$, 
can be achieved in general. 
Developing systematic approaches for handling these higher cases -- 
possibly through deeper algebraic number theory or computational methods -- 
remains an important task. 
Second, efficient algorithms for determining cyclotomic intersections would enable
the verification of larger examples and broaden the known classes of EDFs. 

\item Systematic strong AMD codes: 

\ \ \ We constructed three classes of bounded generalized strong external difference families
via cyclotomy over finite fields, direct products of finite fields, and residue class rings of integers.
These constructions yield many infinite families of systematic G-optimal strong AMD codes.

\ \ \ However, relatively little is known about bounded generalized strong external difference families
compared with difference families and external difference families.
The development of general construction methods for BGSEDFs remains a challenging open problem.
\end{enumerate}

Compared with error-correcting codes, our understanding of algebraic manipulation detection codes remains relatively limited, especially regarding explicit constructions of optimal AMD codes.
Further advances in explicit constructions of optimal AMD codes are needed.

\section*{Acknowledgments}
The authors thank Prof. Ying Miao (University of Tsukuba) for many valuable discussions. 


 
%
\bibliographystyle{IEEEtran}
\bibliography{MFBib.bib}

\end{document}